\documentclass{LMCS}

\def\doi{8 (2:09) 2012}
\lmcsheading%
{\doi}
{1--30}
{}
{}
{Nov.~20, 2011}
{Jun.~\phantom04, 2012}
{}

\usepackage{enumerate}
\usepackage{hyperref}
\usepackage{stmaryrd}
\usepackage{verbatim}
\usepackage{fancyvrb}

\usepackage[all]{hypcap}
\hypersetup{pdftitle=Refining Inductive Types,
            pdfauthor=Robert Atkey and Patricia Johann and Neil Ghani}

\usepackage{amsmath}
\usepackage{amssymb}
\usepackage{stmaryrd}
\usepackage[all]{xy}

\newcommand{\tyname}[1]{\texttt{#1}}
\newcommand{\consname}[1]{\texttt{#1}}
\newcommand{\neutral}[1]{\texttt{#1}}

\newcommand{\inn}{\mathit{in}}
\newcommand{\cat}[1]{\mathcal{#1}}
\newcommand{\sepbar}{\mathrel|}
\newcommand{\fold}[1]{\llparenthesis #1 \rrparenthesis}

\newcommand{\Fam}{\mathrm{Fam}}

\newcommand{\Set}{\mathrm{Set}}
\newcommand{\Alg}{\mathrm{Alg}}
\newcommand{\palg}{U^{\mathsf{Alg}}}

\newcommand{\reindAlg}{*{\mathsf{Alg}}}
\newcommand{\truthalg}{\top^{\mathsf{Alg}}}
\newcommand{\compralg}[1]{\{#1\}^{\mathsf{Alg}}}
\newcommand{\ok}{\mathsf{ok}}
\newcommand{\fail}{\mathsf{fail}}

\newtheorem{eorollary}{Example}

\newcommand{\parenref}[1]{\hyperref[#1]{(\ref*{#1})}}
\newcommand{\lemref}[1]{\hyperref[#1]{Lemma \ref*{#1}}}
\newcommand{\thmref}[1]{\hyperref[#1]{Theorem \ref*{#1}}}
\newcommand{\propref}[1]{\hyperref[#1]{Proposition \ref*{#1}}}
\newcommand{\corref}[1]{\hyperref[#1]{Corollary \ref*{#1}}}
\newcommand{\conref}[1]{\hyperref[#1]{Conjecture \ref*{#1}}}
\newcommand{\exref}[1]{\hyperref[#1]{Example \ref*{#1}}}

\newcommand{\pullbackcorner}[1][dr]{\save*!/#1-1.2pc/#1:(-1,1)@^{|-}\restore}

\newcommand{\adjunction}[4]{\xymatrix{ {#1} \ar@/_/[r]_{#4}
    \ar@{}[r]|{\perp}& \ar@/_/[l]_{#3} {#2} }}

\begin{document}

\title[Refining Inductive Types]{Refining Inductive Types}

\author[R.~Atkey]{Robert Atkey}
\address{University of Strathclyde, UK}
\email{\{Robert.Atkey,Patricia.Johann,Neil.Ghani\}@cis.strath.ac.uk}

\author[P.~Johann]{Patricia Johann}
%\address{University of Strathclyde, UK}
%\email{Patricia.Johann@cis.strath.ac.uk}

\author[N.~Ghani]{Neil Ghani}
%\address{University of Strathclyde, UK}
%\email{Neil.Ghani@cis.strath.ac.uk}


\begin{abstract}
  Dependently typed programming languages allow sophisticated
  properties of data to be expressed within the type system. Of
  particular use in dependently typed programming are indexed types
  that refine data by computationally useful information.  For
  example, the $\mathbb{N}$-indexed type of vectors refines lists by
  their lengths. Other data types may be refined in similar ways, but
  programmers must produce purpose-specific refinements on an {\em ad
    hoc} basis, developers must anticipate which refinements to
  include in libraries, and implementations must often store redundant
  information about data and their refinements. In this paper we show
  how to generically derive inductive characterisations of refinements
  of inductive types, and argue that these characterisations can
  alleviate some of the aforementioned difficulties associated with
  {\em ad hoc} refinements. Our characterisations also ensure that
  standard techniques for programming with and reasoning about
  inductive types are applicable to refinements, and that refinements
  can themselves be further refined.
\end{abstract}

\maketitle

\section{Introduction}\label{sec:introduction}

One of the key aims of current research in functional programming is
to reduce the {\em semantic gap} between what programmers know about
computational entities and what the types of those entities can
express. One particularly promising approach to closing this gap is to
{\em index} types by extra information that can be used to express
properties of their elements. For example, most functional languages
support a standard list data type parameterised over the type of the
data lists contain, but for some applications it is also convenient to
be able to state the length of a list in its type. This makes it
possible, for instance, to ensure that the list argument to the
\verb|tail| function has non-zero length --- i.e., is non-empty ---
and that the lengths of the two list arguments to \verb|zip| are the
same. Without this kind of static enforcement of preconditions,
functions such as these must be able to signal erroneous arguments ---
perhaps using an error monad, or a built-in exception facility --- and
their clients must be able to handle the cases in which an error is
raised.

A data type that equips each list with its length can be defined in
the dependently typed language Agda 2~\cite{agda10} using the
following declaration:
\begin{verbatim}
data Vector (B : Set) : Nat -> Set where
  nil  : Vector B zero
  cons : {n : Nat} -> B -> Vector B n -> Vector B (succ n)
\end{verbatim}

\noindent
This declaration\footnote{The \texttt{\{X : S\}} notation indicates
  that there is an implicit parameter of type \texttt{S}, named
  \texttt{X}. When applying a function with an implicit argument, Agda
  2 will attempt to infer a suitable value for it.} inductively
defines, for each choice of element type \verb|B|, a data type
\verb|Vector B| that is indexed by natural numbers and has two
constructors: $\consname{nil}$, which constructs a vector of data with
type \verb|B| of length zero (here represented by the data constructor
$\consname{zero}$ for the natural numbers), and $\consname{cons}$,
which constructs from an index $\neutral{n}$, an element of
$\tyname{B}$, and a vector of data with type \verb|B| of length
$\neutral{n}$, a new vector of data with type \verb|B| of length
$\neutral{n} + 1$ (here represented by the application
$\consname{succ}\ \neutral{n}$ of the data constructor
$\consname{succ}$ for the natural numbers to $\neutral{n}$). The
inductive type \verb|Vector B| can be used to define functions on
lists with elements of type \verb|B| that are ``length-aware'' in a
way that functions processing data of standard list types cannot
be. For example, it allows length-aware \verb|tail| and \verb|zip|
functions to be given via the following Agda 2 types and definitions:

\begin{verbatim}
tail : {B : Set} -> {n : Nat} -> Vector B (succ n) -> Vector B n
tail (cons b bs) = bs

zip  : {B C : Set} -> {n : Nat} -> 
         Vector B n -> Vector C n -> Vector (B x C) n
zip  nil         nil        = nil
zip (cons b bs) (cons c cs) = cons (b , c) (zip bs cs)
\end{verbatim}

\noindent
Examples such as these suggest that indexing types by computationally
relevant information has great potential. However, for this potential
to be realised we must better understand how indexed types can be
constructed. Moreover, since we want to ensure that all the techniques
that have been developed for structured programming with and
principled reasoning about inductive types\footnote{Recall that an
  inductive data type is one that an be represented as the least fixed
  point $\mu F$ of a functor $F$ on a category suitable for
  interpreting the types in a language.} --- such as those championed
in the Algebra of Programming~\cite{bdm97} literature --- are
applicable to the resulting indexed types, we also want these types to
be inductive.  This paper therefore asks the following fundamental
question:

\vspace*{0.1in}

\begin{quote}
  \em Can elements of one inductive type be systematically augmented
  with computationally relevant information to construct an indexed
  inductive type that captures the computationally relevant
  information in their indices? If so, how?
\end{quote}

\vspace*{0.1in}

\noindent
That is, how can we {\em refine} an inductive type to get a new type,
called a {\em refinement}, that associates to each element of the
original type its index, and how can we ensure that the resulting
refinement is inductive?

\subsection{A Naive Solution}

One straightforward way to refine an inductive type is to use a
refinement function to compute the index for each of its elements and
then to associate these indices to their corresponding elements. To
refine lists by their lengths, for example, we would start with the
standard list data type, which has the following Agda 2
declaration\footnote{Agda 2 allows overloading of constructor names,
  so we reuse the constructor names \texttt{nil} and \texttt{cons}
  from the \texttt{Vector} type defined above.}:
\begin{verbatim}
data List (B : Set) : Set where
  nil  : List B
  cons : B -> List B -> List B
\end{verbatim}
We would then define the following function \verb|length| by
structural recursion on elements of \texttt{List B}\,:
\begin{verbatim}
length : {B : Set} -> List B -> Nat
length nil         = zero
length (cons _ bs) = succ (length bs)
\end{verbatim}
From these we would construct the following refinement of lists by the
function \verb|length|, using a subset type: 
\begin{equation}\label{eqn:vector}
  \texttt{ListWithLength}\ \texttt{B}\ \texttt{n} \cong \{ \texttt{bs} :
  \texttt{List B}
  \sepbar \texttt{length}\ \texttt{bs} = \texttt{n} \}
\end{equation}
(alternatively, we could have also used a $\Sigma$-type to hold the
list $\texttt{bs}$ and the proof that $\texttt{length}\ \texttt{bs} =
\texttt{n}$.)  Note that this construction is {\em global} in that
both the data type and the collection of indices exist {\em a priori},
and the refinement is obtained by assigning, {\em post facto}, an
appropriate index to each data type element. It also suffers from a
serious drawback: the resulting refinement --- \verb|ListWithLength B|
here --- is not presented as an inductive type, so the naive solution
is not a solution to the fundamental question posed above. (In
addition, the refinement \verb|ListWithLength B| does not obviously
have anything to do with the type \texttt{Vector B}.) So the question
remains: how do we get the inductive type \texttt{Vector B} from the
inductive type \texttt{List B}?

\subsection{A Better Solution}\label{sec:better-solution}

When the given refinement function is computed by structural recursion
(i.e., by the fold) over the data type to be refined --- as is the
case for the function \verb|length| above and is often the case in
practice --- then we can give an alternative construction of
refinements that provides a comprehensive answer to the fundamental
question raised above.  In this case we can construct, for each
inductive type $\mu F$ and each $F$-algebra $\alpha$ whose fold
computes the desired refinement function, a functor $F^\alpha$ whose
least fixed point $\mu F^\alpha$ is the desired refinement. This
construction is the central contribution of the paper.  Our
characterisation of the refinement of $\mu F$ by $\alpha$ as the
inductive type $\mu F^\alpha$ allows the entire arsenal of structured
programming techniques based on initial algebras to be brought to bear
on the resulting refinement. By contrast with the construction
in \parenref{eqn:vector} above, our characterisation is also {\em
  local}, in that the indices of recursive substructures are readily
available {\em at the time a structurally recursive program is
  written}, rather than needing to be computed by inversion at run
time from the index of the input data structure to the program.

For each functor $F$ and $F$-algebra $\alpha$, the functor
$F^{\alpha}$ that we construct is intimately connected with the
generic structural induction rule for the inductive type $\mu F$, as
presented by Hermida and Jacobs \cite{hermida98structural} and by
Ghani, Johann, and Fumex \cite{ghani10induction}. This is perhaps not
surprising: structural induction proves properties of functions
defined by structural recursion on elements of inductive types. If the
values of such functions are abstracted into the indices of associated
indexed inductive types, then their computation need no longer be
performed during inductive proofs. In essence, work has been shifted
away from computation and onto data. Refinement can thus be seen as
supporting reasoning by structural induction ``up to'' the index of a
term.

\subsection{The Structure of this Paper}

The remainder of this paper is structured as follows. In
\autoref{sec:f-algebras} we introduce inductive types and recall their
representation as carriers of initial algebras of functors. We first
recall that, for any functor $F$, the collection of $F$-algebras
forms a category, and then give a key theorem due to Hermida and
Jacobs~\cite{hermida98structural} relating different $F$-algebras and,
thereby, different refinements of $\mu F$. In \autoref{sec:families}
we define the fibrational framework for refinements with which we work
in this paper, and introduce the important idea of the lifting of a
functor. In \autoref{sec:refining-inductive} we show how liftings can
be used to refine inductive types, prove the correctness of our
construction of refinements, and illustrate our construction with some
simple examples.  In \autoref{sec:indexed-refinement} we show how to
refine inductive types that are themselves already indexed, thus
extending our construction to allow refinement of the whole range of
indexed inductive types available in dependently typed languages.  In
\autoref{sec:partial-refinement} we further extend our basic
refinement technique to allow partial refinement, in which indexed
types are constructed from inductive types not all of whose elements
have indices. Our motivating example here is that of expressions that
can fail to be well-typed. Indeed, we refine the type of possibly
ill-typed expressions by a type checker to yield the indexed inductive
type of well-typed expressions. In \autoref{sec:zygo-refine} we extend
the basic notion of refinement in yet another direction to allow
refinement by paramorphisms --- also known as primitive recursive
functions --- and their generalisation zygomorphisms. Perhaps
surprisingly, this takes us from the world of indexed inductive types
to indexed induction-recursion, in which inductive types and recursive
functions are defined simultaneously.  In \autoref{sec:discussion} we
conclude and discuss related and future work.

Throughout this paper, we adopt a semantic approach based on category
theory because it allows a high degree of abstraction and
economy. More specifically, we develop our theory in the abstract
setting of fibrations~\cite{jacobs99book}. Nevertheless, we specialise
to the families fibration over the category of sets in order to
improve accessibility and give concrete intuitions;
\autoref{sec:families} gives the necessary definitions and
background. Moreover, carefully using only the abstract structure of
the families fibration allows us to expose crucial structure that
might be lost were a specific programming notation to be used. This
structure both simplifies our proofs and facilitates the iteration of
our construction detailed in \autoref{sec:indexed-refinement}. It also
highlights the commonalities between the various constructions we
present. In particular, each of the refinement processes we discuss
produces functors of the form $J \circ \hat{F}$, where $\hat{F}$ is
the lifting of the functor $F$ defining the data type $\mu F$ to be
refined. We are currently investigating whether this observation leads
to a more general theory of refinement, as well as its potential use
in structuring an implementation. A type-theoretic, rather than
categorical, answer to the fundamental question this paper addresses
has already been given by McBride~\cite{mcbride10ornaments} using his
notion of ornaments for data types (see \autoref{sec:discussion}).

\subsection{Differences from the Previously Published Version}

This paper is a revised and expanded version of the FoSSaCS 2011
conference version \cite{atkey11refinement}. Additional explanations
have been provided throughout, examples have been expanded, and some
of the material has been reordered for
clarity. \autoref{sec:indexed-ind-types}, which explains in more
detail the connection between initial algebras and the indexed
inductive types present in systems such as Agda 2, is entirely
new. \autoref{sec:zygo-refine}, which discusses the connection between
refinement by zygomorphisms and indexed inductive-recursive
definitions, is also completely new, and represents significant
further development of our basic refinement technique.

\section{Inductive Types and $F$-algebras}\label{sec:f-algebras}

A data type is {\em inductive (in a category $\cat{C}$)} if it is the
least fixed point $\mu F$ of an endofunctor $F$ on $\cat{C}$, in a
sense to be made precise in \autoref{sec:f-algs} below. For example,
if $\Set$ denotes the category of sets and functions, $\mathbb{Z}$ is
the set of integers, and $+$ and $\times$ denote the coproduct and
product, respectively, then $\mu F_{\tyname{Tree}}$ for the
endofunctor $F_{\tyname{Tree}}X = \mathbb{Z} + X \times X$ on $\Set$
represents the following data type of binary trees with integer data
at the leaves:
\begin{verbatim}
data Tree : Set where
    leaf : Integer -> Tree
    node : (Tree x Tree) -> Tree
\end{verbatim}

\subsection{$F$-algebras}\label{sec:f-algs}

Our precise understanding of inductive types comes from the
categorical notion of an $F$-algebra. If $\cat{C}$ is a category and
$F$ is an endofunctor on $\cat{C}$, then an \emph{$F$-algebra} is a
pair $(A, \alpha : FA \to A)$ comprising an object $A$ of $\cat{C}$
and a morphism $\alpha : FA \to A$ in $\cat{C}$. The object $A$ is
called the \emph{carrier} of the $F$-algebra, and the morphism
$\alpha$ is called its \emph{structure map}. We usually refer to an
$F$-algebra solely by its structure map $\alpha : FA \to A$, since the
carrier is present in the type of this map.

An \emph{$F$-algebra morphism} from $\alpha : FA \to A$ to $\alpha' :
FB \to B$ is a morphism $f : A \to B$ of $\cat{C}$ such that $f \circ
\alpha = \alpha' \circ Ff$. An $F$-algebra $\alpha : FA \to A$ is
\emph{initial} if, for any $F$-algebra $\alpha' : FB \to B$, there
exists a unique $F$-algebra morphism from $\alpha$ to $\alpha'$. If it
exists, the initial $F$-algebra is unique up to isomorphism, and
Lambek's Lemma further ensures that the\footnote{We identify
  isomorphic entities when convenient. When doing so, we write $=$ in
  place of $\cong$.}  initial $F$-algebra is an isomorphism. Its
carrier is thus the least fixed point $\mu F$ of $F$. We write $\inn_F
: F(\mu F) \to \mu F$ for the initial $F$-algebra, and
$\fold{\alpha}_F : \mu F \to A$ for the unique morphism from $\inn_F :
F(\mu F) \to \mu F$ to any $F$-algebra $\alpha : FA \to A$. We write
$\fold{-}$ for $\fold{-}_F$ when $F$ is clear from context.  Of
course, not all functors have initial algebras. For instance, the
functor $FX = (X \to 2) \to 2$ on $\Set$ does not have an initial
algebra.

In light of the above, the data type $\tyname{Tree}$ can be
interpreted as the carrier of the initial
$F_{\tyname{Tree}}$-algebra. In functional programming terms, a
function $\alpha : \mathbb{Z} + A \times A \to A$ is an
$F_{\tyname{Tree}}$-algebra, and the function $\fold{\alpha} :
\tyname{Tree} \to A$ induced by the initiality property is exactly the
application to $\alpha$ of the standard iteration function
$\verb|fold|$ for trees (actually, the application of \verb|fold| to
an ``unbundling'' of $\alpha$ into replacement functions, one for each
of $F_{\tyname{Tree}}$'s constructors). More generally, for each
functor $F$, the function $\fold{-}_F : (F A \to A) \to \mu F \to A$
is the standard iteration function for $\mu F$.

\subsection{Indexed Inductive Types as
  $F$-Algebras}\label{sec:indexed-ind-types}  

Indexed types can be inductive, and this gives rise to the notion of
an {\em indexed inductive type}. Such a type is also called an {\em
  inductive family} of types~\cite{dybjer94inductive}. Indexed
inductive types can be seen as initial $F$-algebras for endofunctors
$F$ on categories of indexed sets. For example, if $B$ is a set of
elements, then we can define a functor $F_{\tyname{Vector}_B}$ on the
category of $\mathbb{N}$-indexed sets whose least fixed point
represents the inductive data type $\tyname{Vector B}$ from
\hyperref[sec:introduction]{introduction}. The two constructors
$\texttt{nil}$ and $\texttt{cons}$ are reflected in the definition of
$F_{\tyname{Vector}_B}$ as a coproduct, the individual arguments to
each constructor are reflected as products within each summand of this
coproduct, and the implicit equality constraints on the indices are
reflected as explicit equality constraints. We define
\begin{eqnarray*}
F_{\tyname{Vector}_B} & : & (\mathbb{N} \to \Set) \to
(\mathbb{N} \to \Set)\\ 
 F_{\tyname{Vector}_B}\; X \;
 &=& \lambda n. \{* \sepbar n = 0\} + \{ (n_1 : \mathbb{N}, a : B, x : X n_1) \sepbar n = n_1 + 1 \}
\end{eqnarray*}
where the notation $\{* \sepbar n = 0\}$ denotes the set $\{*\}$ when
$n = 0$ and the empty set otherwise.  The carrier of the initial
algebra $\inn_{F_{\tyname{Vector}_B}} : F_{\tyname{Vector}_B}(\mu
F_{\tyname{Vector}_B}) \to \mu F_{\tyname{Vector}_B}$ of this functor
consists of the $\mathbb{N}$-indexed family $\mu
F_{\tyname{Vector}_B}$ of sets of vectors with elements from $B$,
together with a function $\inn_{F_{\tyname{Vector}_B}}$ that ``bundles
together'' the constructors $\texttt{nil}$ and $\texttt{cons}$. In
\autoref{sec:refining-examples} below we show how
$F_{\tyname{Vector}_B}$ can be \emph{derived} from the functor
$F_{\tyname{List}_B}$ whose least fixed point is the inductive type of
lists with elements from $B$, together with the algebra
$\mathit{lengthalg}$ whose fold is the standard length function on
lists.

In general, $X$-indexed inductive types can be understood as initial
algebras of functors $F : (X \to \Set) \to (X \to \Set)$. In
\autoref{sec:families} below we will see how the collection of
categories of indexed sets can be organised into the \emph{families
  fibration}, in which we carry out the constructions giving rise to
our framework for refinement.

\subsection{Categories of $F$-algebras}

If $F$ is an endofunctor on $\cat{C}$, we write $\Alg_F$ for the
category whose objects are $F$-algebras and whose morphisms are
$F$-algebra morphisms between them. Identities and composition in
$\Alg_F$ are taken directly from $\cat{C}$. The existence of initial
$F$-algebras is equivalent to the existence of initial objects in the
category $\Alg_F$.

In Theorems~\ref{thm:initial-fhat-algebra}
and~\ref{thm:partial-refinement} below, we will have an initial object
in one category of algebras and want to show that applying a functor
to it gives the initial object in another category of algebras. We
will use adjunctions to do this. Recall that an {\em adjunction}
$\adjunction{\cat{C}}{\cat{D}}{L}{R}$ between two categories $\cat{C}$
and $\cat{D}$ consists of a left adjoint functor $L$, a right adjoint
functor $R$, and an isomorphism natural in $A$ and $X$ between the set
$\cat{C}(LA,X)$ of morphisms in $\cat{C}$ from $LA$ to $X$ and the set
$\cat{D}(A,RX)$ of morphisms in $\cat{D}$ from $A$ to $RX$. We say
that the functor $L$ is {\em left adjoint} to $R$, and that the
functor $R$ is {\em right adjoint} to $L$, and we write $L \dashv
R$. To lift adjunctions to categories of algebras, we will make much
use of the following theorem of Hermida and Jacobs
\cite{hermida98structural}:
\begin{thm}\label{thm:alg-adjunctions}
  If $F : \cat{C} \to \cat{C}$ and $G : \cat{D} \to \cat{D}$ are
  functors, $L \dashv R$, and $F \circ L \cong L \circ G$ is a natural
  isomorphism, then the adjunction
  $\adjunction{\cat{C}}{\cat{D}}{L}{R}$ lifts to an adjunction
  $\adjunction{\Alg_F}{\Alg_G}{L'}{R'}$.
\end{thm}

\noindent
In the setting of the theorem, if $G$ has an initial algebra, then so
does $F$ since left adjoints preserve colimits and in particular
initial objects. To compute the initial $F$-algebra in concrete
situations we need to know that $L' (k : GA \to A) = L k \circ p_A$,
where $p$ is (one half of) the natural isomorphism between $F \circ L$
and $L \circ G$. Then the initial $F$-algebra is given by applying
$L'$ to the initial $G$-algebra, i.e., $\mathit{in}_F =
L'(\mathit{in}_G)$, and hence $\mu F = L' (\mu G)$.

\section{A Framework for Refinement}\label{sec:families}

We develop our theoretical framework for refinement in the setting of
fibrational models of extensional Martin-L\"of type theory, which is a
key theory underlying dependently typed programming.  Since the
concepts and terminology of fibrational category theory will not be
familiar to most readers, we have taken care to formulate each of our
definitions and main theorems in the families fibration. The families
fibration gives the archetypal semantics of Martin-L\"of type theory,
in which indexed types are interpreted directly as indexed sets. In
this section we define the families fibration, and identify the parts
of its structure that we require for the rest of the paper. As readers
who are familiar with the categorical notion of fibration will
observe, the terminology and structure that we identify comes from
fibred category theory. We take care to identify the particular
properties of the families fibration that are required for our results
to hold, and refer to the literature for the formulation of these
properties in the general setting.

\subsection{The Families Fibration}

As is customary, we model indexed types in the category $\Fam(\Set)$.
An object of $\Fam(\Set)$ is a pair $(A,P)$ comprising a set $A$ and a
function $P : A \to \Set$; such a pair is called a {\em family} of
sets. We denote a family $(A,P)$ as $P : A \to \Set$ when convenient,
or simply as $P$ when $A$ can be inferred from context. A morphism
$(f,f^{\sim}) : (A,P) \to (B,Q)$ of $\Fam(\Set)$ is a pair of
functions $f : A \to B$ and $f^{\sim} : \forall a.\ P a \to Q(f
a)$. From a programming perspective, a family $(A,P)$ is an
$A$-indexed type $P$, where $P a$ represents the collection of data
with index $a$. An alternative, logical, view is that $(A,P)$ is a
predicate representing a property $P$ of data of type $A$, and that $P
a$ represents the collection of proofs that $a$ has property $P$. When
$P a$ is inhabited, $P$ is said to {\em hold} for $a$. When $P a$ is
empty, $P$ is said {\em not to hold} for $a$. We will freely switch
between the programming and logical interpretations of families when
providing intuition for our formal development below.

The {\em families fibration} $U : \Fam(\Set) \to \Set$ is the functor
mapping each family $(A,P)$ to $A$ and each morphism $(f,f^{\sim})$ to
$f$. The category $\Set$ is referred to as the {\em base category} of
the families fibration and $\Fam(\Set)$ is referred to as its {\em
  total category}. For each set $A$, the category $\Fam(\Set)_A$
consists of families $(A, P)$ and morphisms $(f, f^\sim)$ between them
such that $f = \mathit{id}_A$. Such a morphism is said to be a {\em
  vertical morphism}. Similarly, a {\em vertical natural
  transformation} is a natural transformation each of whose components
is a vertical morphism. We say that an object or morphism in
$\Fam(\Set)_A$ is {\em over} $A$ with respect to the families
fibration, and call $\Fam(\Set)_A$ the \emph{fibre} of the families
fibration over $A$.  A function $f : A \to B$ contravariantly
generates a \emph{reindexing functor} $f^* : \Fam(\Set)_B \to
\Fam(\Set)_A$ for the families fibration that maps $(B, Q)$ to $(A, Q
\circ f)$.

\subsection{Truth and Comprehension}\label{sec:truth-compr}

Each fibre $\Fam(\Set)_A$ has a terminal object $(A, \lambda a:A.\
1)$, where $1$ is the canonical singleton set.  In light of the
logical reading of families above, this object is called the {\em
  truth predicate} for $A$. The mapping of objects to their truth
predicates extends to a functor $\top : \Set \to \Fam(\Set)$, called
the {\em truth functor} for the families fibration.  In addition, for
each family $(A,P)$ we can define the {\em comprehension} of $(A,P)$,
denoted $\{(A,P)\}$, to be $\Sigma a \!:\! A. Pa$, i.e., $\{ (a,p)
\sepbar a \in A, p \in P a \}$. The comprehension $\{(A,P)\}$ packages
elements $a \in A$ with proofs $p \in Pa$.  The mapping of families to
their comprehensions extends to a functor $\{-\}: \Fam(\Set) \to
\Set$, called the {\em comprehension functor} for the families
fibration. Overall, we have the following pleasing collection of
adjoint relationships:

\begin{equation}\label{ex:cc-with-unit}
  \xymatrix{
    {\Fam(\Set)}
    \save[]+<0.2cm,-0.66cm>*{\footnotesize\dashv}\restore
    \save[]+<1cm,-0.66cm>*{\footnotesize\dashv}\restore
    \ar[d]_U
    \ar@/^3pc/[d]^{\{-\}}    
    \\
    {\Set}
    \ar@/_1pc/[u]_{\top}
  }
\end{equation}
The families fibration $U$ is thus a \emph{comprehension category with
  unit} \cite{jacobs93comprehension,jacobs99book}. Like every
comprehension category with unit, $U$ supports a natural
transformation $\pi : \{-\} \to U$ such that $\pi_{(A,P)}(a,p) = a$
for all $(a,p)$ in $\{(A,P)\}$, projecting out the $A$ component from
a comprehension. In fact, $U$ is a {\em full comprehension category
with unit}, i.e., the functor from $\Fam(\Set)$ to $\Set^\to$ induced
by $\pi$ is full and faithful. Here, $\Set^\to$ is the {\em arrow
  category} of $\Set$. Its objects are morphisms of $\Set$ and its
morphisms from $f:X\to Y$ to $f':X'\to Y'$ are pairs $(\alpha_1,
\alpha_2)$ of morphisms in $\Set$ such that $f' \circ \alpha_1 =
\alpha_2 \circ f$. Fullness means that the action of $\pi$ on
morphisms is surjective, and faithfulness means that it is injective.
Fullness will be used in the proof of \thmref{thm:change-of-base}
below, when we consider refinements of indexed types.

\subsection{Indexed Coproducts}\label{sec:sums}

For each function $f : A \to B$ and family $(A,P)$, we can form the
family $\Sigma_f(A,P) = (B, \lambda b.\ \Sigma_{a \in A}.\ (b = f a)
\times P a)$, called the {\em indexed coproduct of $(A,P)$ along $f$}.
The mapping of each family to its indexed coproduct along $f$ extends
to a functor $\Sigma_f : \Fam(\Set)_A \to \Fam(\Set)_B$ which is left
adjoint to the reindexing functor $f^*$ for the families fibration.
In the abstract setting of fibrations, a fibration with the property
that each re-indexing functor $f^*$ has a left adjoint $\Sigma_f$ is
called a \emph{bifibration}, and the functors $\Sigma_f$ are called
{\em op-reindexing} functors. A bifibration that is also a full
comprehension category with unit is called a {\em full cartesian
  Lawvere category}~\cite{jacobs93comprehension}. The families
fibration is a full cartesian Lawvere category.

The functors $\Sigma_f$ are often subject to the Beck-Chevalley
condition for coproducts, which is well-known to hold for the families
fibration. This condition ensures that, in certain circumstances,
op-reindexing commutes with re-indexing~\cite{jacobs99book}. It is
used in the proof of \lemref{lem:lifting-commute-reindex}.

At several places below we make essential use of the fact that the
families fibration has very strong coproducts, i.e., that in the diagram
\begin{equation}\label{eqn:sums-compr}
  \xymatrix{
    {\{(A,P)\}}
    \ar[r]^(.4){\{\psi\}}
    \ar[d]_{\pi_{(A,P)}}
    &
    {\{\Sigma_f (A,P) \}}
    \ar[d]^{\pi_{\Sigma_f (A,P)}}
    \\
    {A}
    \ar[r]^{f}
    &
    {B}
  }
\end{equation}
\noindent
where $\psi$ is the obvious map of families of sets over $f$,
$\{\psi\}$ is an isomorphism. This notion of very strong coproducts
naturally generalises the usual notion of strong
coproducts~\cite{jacobs99book}, and imposes a condition that is
standard in models of type theory.

\subsection{Indexed Products}\label{sec:indexed-products}

For each function $f : A \to B$ and family $(A,P)$ we can also form
the family $\Pi_f(A,P) = (B, \lambda b.\ \Pi_{a \in A}. (b = f a) \to
P a)$, called the {\em indexed product of $(A,P)$ along $f$}.  The
mapping of each family to its indexed product along $f$ extends to a
functor $\Pi_f : \Fam(\Set)_A \to \Fam(\Set)_B$ which is right adjoint
to the reindexing functor $f^*$ for the families fibration. Altogether
we have the following collection of relationships for each function $f
: A \to B$:
\begin{displaymath}
  \xymatrix{
    {\Fam(\Set)_B}
    \save[]+<1.35cm,0.3cm>*{\scriptsize{\perp}}\restore
    \save[]+<1.35cm,-0.25cm>*{\scriptsize{\perp}}\restore
    \ar[r]|{f^*}
    &
    {\Fam(\Set)_A}
    \ar@/_1.2pc/[l]_{\Sigma_f}
    \ar@/^1pc/[l]^{\Pi_f}    
  }
\end{displaymath}

\noindent
Like its counterpart for indexed coproducts, the Beck-Chevalley
condition for indexed products is often required and indeed it holds
in the families fibration. We do not make use of this condition in
this paper.

\subsection{Liftings}\label{sec:fibrational-ind}

The relationship between inductive types and their refinements can be
given in terms of liftings of functors. A {\em lifting} of a functor
$F : \Set \to \Set$ is a functor $\hat{F} : \Fam(\Set) \to \Fam(\Set)$
such that $F \circ U = U \circ \hat{F}$.  A lifting is {\em
  truth-preserving} if there is a natural isomorphism $\top \circ F
\cong \hat{F} \circ \top$. Truth-preserving liftings for all
polynomial functors --- i.e., for all functors built from identity
functors, constant functors, coproducts, and products --- were given
by Hermida and Jacobs~\cite{hermida98structural}.  Truth-preserving
liftings were established for arbitrary functors by Ghani \emph{et
  al.}~\cite{ghani10induction}.  Their truth-preserving lifting
$\hat{F}$ is defined on objects by
\begin{equation}
  \label{eqn:fhat}
  \hat{F}(A,P)
  \begin{array}[t]{cl}
    = & (FA, \lambda x.\ \{ y : F \{ (A,P)\} \sepbar F\pi_{(A,P)} y = x \}) \\
    = & \Sigma_{F\pi_{(A,P)}} \top(F \{(A,P)\})
  \end{array}
\end{equation}

\noindent
Reading this definition logically, we can say that $\hat{F}(A,P)$
holds for $x \in FA$ if $P$ holds for every $a \in A$ ``inside''
$x$. Thus $\hat{F}$ is a generic definition of the \emph{everywhere}
modality, as defined for containers by Altenkirch and Morris
\cite{alten09indexed}. This can be seen clearly by considering the
action of the lifting in \parenref{eqn:fhat} on polynomial functors:
\[\begin{array}{lcl}
  \widehat{Id}(A, P) &=&(A,P)\\
  \widehat{K_B}(A, P) &=&\top B = (B, \lambda x.\ 1)\\
  \widehat{(F+G)}(A, P) & = & \left(FA+GA, \ \lambda a. \mathsf{case}\ a\ \mathsf{ of}
    \left\{
   \begin{array}{l}
     \mathsf{inl}\ x \Rightarrow \hat{F}(A,P)x \\
     \mathsf{inr}\ y \Rightarrow \hat{G}(A,P)y
   \end{array}
 \right.\right)\\
 \widehat{(F \times G)}(A, P)
 & = &  (FA \times GA, \ \lambda (a,b).\ \hat{F}(A, P)a \times
 \hat{G}(A, P)b)  
\end{array}\]
The identity functor on $\Set$ does not contribute any new information
to proofs that a property holds for a given data element, so its
lifting is the identity functor on $\Fam(\Set)$. For any $B$, the
constantly $B$-valued functor $K_B$ on $\Set$ does not contribute any
inductive information to proofs, so its lifting is the truth predicate
$\top B$ for $B$. The lifting of a coproduct of functors splits into
two possible cases, depending on the value being analysed. And a
product of functors contributes proof information from each of its
components. Lifting is defined generically in terms of the functor
$F$, and so it is possible to compute the lifting of non-polynomial
functors such as the the finite powerset functor. Ghani, Johann and
Fumex \cite{ghani10induction} give further examples of the lifting
$\hat{F}$ applied to non-polynomial functors.

Below, in \hyperref[lem:lifting-commute-reindex]{Lemmas
  \ref*{lem:lifting-commute-reindex}} and
\ref{lem:lifting-commute-sigma} and
\hyperref[sec:refining-inductive]{Sections
  \ref*{sec:refining-inductive}}, \ref{sec:indexed-refinement},
\ref{sec:partial-refinement} and \ref{sec:zygo-refine}, we will be
interested in the restriction of the lifting $\hat{F}$ to fibres over
particular sets $A$. Given an object $(A,P)$ of $\Fam(\Set)_A$,
$\hat{F}(A,P)$ is an object of $\Fam(\Set)_{FA}$. Therefore, if we
restrict the domain of $\hat{F}$ to $\Fam(\Set)_A$, we get a functor
$\hat{F}_A : \Fam(\Set)_A \to \Fam(\Set)_{FA}$. The subscript $A$ on
$\hat{F}_A$ indicates that we have restricted the domain to
$\Fam(\Set)_A$.

The final expression in \parenref{eqn:fhat} is given in terms of the
constructions of \hyperref[sec:truth-compr]{Sections
  \ref*{sec:truth-compr}} and \hyperref[sec:sums]{ \ref*{sec:sums}},
so the definition of a lifting makes sense in any full cartesian
Lawvere category.

Under certain conditions, the lifting $\hat{F}$ for any functor $F$ is
well-behaved with respect to reindexing and op-reindexing. We make
this observation precise in two lemmas that will be used in our
development of both our basic (\autoref{sec:refining-inductive}) and
partial refinement techniques (\autoref{sec:partial-refinement}). To
state the first, we need the notion of a pullback; this notion will
also be used in \hyperref[sec:indexed-refinement]{Sections
  \ref*{sec:indexed-refinement}}, \hyperref[sec:partial-refinement]{
  \ref*{sec:partial-refinement}}, and \hyperref[sec:zygo-refine]{
  \ref*{sec:zygo-refine}} below. The {\em pullback} of the morphisms
$f : X \to Z$ and $g : Y \to Z$ consists of an object $W$ and two
morphisms $i : W \to X$ and $j : W \to Y$ such that $g \circ j = f
\circ i$.  We indicate pullbacks diagrammatically by
\[\xymatrix{ & W \ar[d]_{i} \ar[r]^{j} \pullbackcorner &
  \ar[d]^g Y \\
  &X \ar[r]_f &Z}\]
Moreover, for any $W'$, $i': W' \to X$, and $j': W' \to Y$ such that
$g \circ j' = f \circ i'$, there exists a unique morphism $u : W' \to
W$ such that $i \circ u = i'$ and $j \circ u = j'$.  When it exists,
the pullback of $f$ and $g$ is unique up to (unique) isomorphism.  All
container functors~\cite{DBLP:conf/fossacs/AbbottAG03}, and hence all
functors modelling strictly positive types, preserve pullbacks.

We can now state our lemmas.

\begin{lem}\label{lem:lifting-commute-reindex}
  For any functor $F : \Set \to \Set$ that preserves pullbacks,
  lifting commutes with reindexing, i.e., for all functions $f : A \to
  B$, there exists a vertical natural isomorphism $\hat{F}_A \circ f^*
  \cong (Ff)^* \circ \hat{F}_B$.
\end{lem}

\begin{lem}\label{lem:lifting-commute-sigma}
  For any functor $F : \Set \to \Set$, lifting commutes with
  op-reindexing, i.e., for all functions $f : A \to B$, there exists a
  vertical natural isomorphism $\hat{F}_B \circ \Sigma_f \cong
  \Sigma_{Ff} \circ \hat{F}_A$.
\end{lem}
\noindent
More generally, \lemref{lem:lifting-commute-reindex} holds in any full
cartesian Lawvere category satisfying the Beck-Chevalley condition for
coproducts, whereas \lemref{lem:lifting-commute-sigma} holds in any
full cartesian Lawvere category with very strong coproducts.

Since $\hat{F}$ is an endofunctor on $\Fam(\Set)$, the category
$\Alg_{\hat{F}}$ of $\hat{F}$-algebras exists. The families fibration
$U : \Fam(\Set) \to \Set$ extends to a fibration $\palg :
\Alg_{\hat{F}} \to \Alg_F$, called the {\em algebras fibration}
induced by $U$. Concretely, the action of $\palg$ is the same as that
of $U$, so that $\palg (k : \hat{F}P \to P) = (Uk : FUP \to UP)$ on
objects and $\palg (h : (k_1 : \hat{F}P \to P) \to (k_2 : \hat{F}Q \to
Q)) = Uh$ on morphisms.  Moreover, writing $\truthalg{}$ and
$\compralg{-}$ for the truth and comprehension functors for $\palg$,
respectively, the adjoint relationships from
\hyperref[ex:cc-with-unit]{Diagram \ref*{ex:cc-with-unit}} all lift to
give $\palg \dashv \truthalg{} \dashv \compralg{-}$.  The two
adjunctions here follow from \thmref{thm:alg-adjunctions} using the
fact that $\hat{F}$ is a truth-preserving lifting. That left adjoints
preserve initial objects can now be used to establish the following
fundamental result, originally from Hermida and Jacobs
\cite{hermida98structural}, and generalised by Ghani \emph{et al.}
\cite{ghani10induction}:
\begin{thm}\label{thm:initial-fhat-algebra}
  $\top(\mu F)$ is the carrier $\mu \hat{F}$ of the initial
  $\hat{F}$-algebra.
\end{thm}
\noindent
\thmref{thm:initial-fhat-algebra} can be generalised to any full
cartesian Lawvere category.  As shown by Hermida and Jacobs, and by
Ghani \emph{et al.}, it can be used to give a generic structural
induction rule for any functor $F$ having an initial algebra.

\section{From Liftings to Refinements}\label{sec:refining-inductive}

In this section we show that the refinement of an
inductive type $\mu F$ by an $F$-algebra $\alpha : FA \to A$ --- i.e.,
the family
\begin{equation}\label{eqn:muFalpha}
  (A, \lambda a:A.\ \{ x : \mu F \sepbar \fold{\alpha} x = a \})
\end{equation}
generalising the refinement in \parenref{eqn:vector} --- is
inductively characterised as $\mu F^\alpha$, where $F^\alpha :
\Fam(\Set)_A \to \Fam(\Set)_A$ is given by
\begin{equation}
  \label{eqn:Falpha-set}
  F^\alpha (A,P) = (A,\lambda a.\ \{ x : F \{(A,P)\} \sepbar
  \alpha(F\pi_{(A,P)} x) = a \}) 
\end{equation}
\noindent
That is, $F^{\alpha}(A,P)$ is obtained by first building the
$FA$-indexed type $\hat{F}(A,P)$ from~\hyperref[eqn:fhat]{Equation
  \ref*{eqn:fhat}}, and then restricting membership to those elements
whose $\alpha$-values are correctly computed from those of their
immediate subterms. More generally, we can express $F^\alpha$ in terms
of the constructions of \hyperref[sec:families]{Section
  \ref*{sec:families}} as 
\begin{equation}
  \label{eqn:Falpha}
  F^\alpha = \Sigma_\alpha \circ \hat{F}_A
\end{equation}

Before we prove that the above construction of $F^\alpha$ is correct,
we show that it yields the refinement of lists by the length function
given in \parenref{eqn:vector}.

\begin{eorollary}\label{ex:lists-n-vectors}
  The inductive type of lists of elements with type $B$ can be
  specified by the functor $F_{\tyname{List}_B}X = 1 + B \times
  X$. Writing $\mathsf{Nil}$ for the left injection and
  $\mathsf{Cons}$ for the right injection into the coproduct
  $F_{\tyname{List}_B}X$, the $F_{\tyname{List}_B}$-algebra
  $\mathit{lengthalg} : F_{\tyname{List}_B}\mathbb{N} \to \mathbb{N}$
  that computes the lengths of lists is
\begin{displaymath}
  \begin{array}{lll}
    \mathit{lengthalg}\ \mathsf{Nil} & = & 0 \\
    \mathit{lengthalg}\ (\mathsf{Cons}(b,n)) & = & n + 1
  \end{array}
\end{displaymath}
In the families fibration, we can calculate the refinement of $\mu
F_{\tyname{List}_B}$ by the algebra $\mathit{lengthalg}$ as follows:
\begin{eqnarray*}
  & & F_{\tyname{List}_B}^{\mathit{lengthalg}}(\mathbb{N}, P) \\
  &=& (\mathbb{N}, \lambda n. \{ x :
  F_{\tyname{List}_B}\{(\mathbb{N},P)\} \sepbar
  \mathit{lengthalg}(F_{\tyname{List}_B}\pi_{(A,P)}x) = n \})\\ 
  &=& (\mathbb{N}, \lambda n.
  \begin{array}[t]{l}
    \{ x : 1 \sepbar \mathit{lengthalg}(\mathsf{Nil}) = n \} \\
    +\\
    \{ x : B \times \{(\mathbb{N}, P)\} \sepbar
    \mathit{lengthalg}(\mathsf{Cons}((B \times
    \pi_{(\mathbb{N},P)})x)) = n \})
  \end{array}
\end{eqnarray*}
The first equality holds by \parenref{eqn:Falpha} and the expansion of
this expression in the families fibration. The second is obtained by
unfolding the definition of $F_{\tyname{List}}$ as a coproduct, which
allows the refinement to be presented as a coproduct as well. In the
first summand of the final expression above,
$\mathit{lengthalg}(\mathsf{Nil}) = 0$, so that $\{ x : 1 \sepbar
\mathit{lengthalg}(\mathsf{Nil}) = n \}$ reduces to $\{ * \sepbar 0 =
n \}$. We can expand the product and comprehension parts of $x$ in the
second summand to see that $\{ x : B \times \{(\mathbb{N}, P)\}
\sepbar \mathit{lengthalg}(\mathsf{Cons}((B \times
\pi_{(\mathbb{N},P)})x)) = n \}$ reduces to $\{ b : B, n_1 :
\mathbb{N}, l : Pn_1 \sepbar \mathit{lengthalg}(\mathsf{Cons}(b,n_1))
= n \}$. Since $\mathit{lengthalg}(\mathsf{Cons}(b,n_1)) = n_1 + 1$,
the whole refinement can therefore be expressed as
\begin{displaymath}
  F_{\tyname{List}_B}^{\mathit{lengthalg}}(\mathbb{N},P) = (\mathbb{N}, \lambda n. \{
  * \sepbar 0 = n \} + \{ b : B, n_1 : \mathbb{N}, l : Pn_1 \sepbar
  n_1 + 1 = n \}) 
\end{displaymath}
As we shall see in \thmref{thm:muFalpha-initial} below, the least
fixed point $\mu F_{\tyname{List}_B}^{\mathit{lengthalg}}$ of this
functor exists and is $(\mathbb{N}, \lambda n.\ \{ x : \mu
F_{\tyname{List}_B} \sepbar \fold{\mathit{lengthalg}}x = n \})$,
exactly as required. Moreover, the expression for
$F_{\tyname{List}_B}^{\mathit{lengthalg}}$ derived here is exactly the
same as the definition of the functor $F_{\tyname{Vector}_B}$ given in
\autoref{sec:indexed-ind-types} whose least fixed point models the
Agda 2 declaration of \texttt{Vector B} given in the
\hyperref[sec:introduction]{introduction}.  The derivation just
completed therefore justifies this definition of \verb|Vector B|.
\end{eorollary}

\subsection{Correctness of Refinement}

We now turn our attention to proving the correctness of our refinement
construction from \parenref{eqn:Falpha-set}. The proof makes good use
of the relationship between the category $\Fam(\Set)$ and the
categories $\Fam(\Set)_A$ for various sets $A$, as well as of the
lifting of this relationship to the categories $\Alg_{\hat{F}}$ and
$\Alg_{F^\alpha}$ of algebras.

We begin with a simple, but key, observation, namely:

\begin{lem}\label{lem:hidden-cartesian}
  Let $(A,P)$ and $(B,Q)$ be objects in $\Fam(\Set)$, and let $f : A
  \to B$ be a function. The set of morphisms $h$ in $\Fam(\Set)$ from
  $(A,P)$ to $(B,Q)$ such that $Uh = f$ is isomorphic to the set of
  morphisms in $\Fam(\Set)_A$ from $(A,P)$ to $f^*(B,Q)$.
\end{lem}
\begin{proof}
  This follows directly from the definitions. On the one hand, a
  morphism $h$ in $\Fam(\Set)$ from $(A,P)$ to $(B,Q)$ such that $Uh =
  f$ is a pair $(f,h^\sim)$, where $h^\sim : \forall a. Pa \to
  Q(fa)$. On the other, the definition of the re-indexing functor
  $f^*$, i.e.  $f^*(B,Q) = (A, Q \circ f)$, entails that a morphism in
  $\Fam(\Set)_A$ from $(A,P)$ to $f^*(B,Q)$ is a pair $(\mathit{id},
  h^\sim)$, where $h^\sim : \forall a. Pa \to Q(fa)$. There is clearly
  an isomorphism between these sets of morphisms.
\end{proof}

To understand the relationship between the category of
$\hat{F}$-algebras and the category of $F^\alpha$-algebras, it is
convenient to define category of $\hat{F}$-algebras that are over the
$F$-algebra $\alpha$ with respect to the fibration $\palg$ defined at
the end of \autoref{sec:fibrational-ind}.

\begin{defi}
  For each $F$-algebra $\alpha : FA \to A$, the category
  $(\Alg_{\hat{F}})_\alpha$ of {\em $\hat{F}$-algebras over $\alpha$}
  with respect to $\palg$ has as objects $\hat{F}$-algebras $k :
  \hat{F}P \to P$ such that $Uk = \alpha$, and as morphisms
  $\hat{F}$-algebra morphisms $f : (k_1 : \hat{F}P \to P) \to (k_2 :
  \hat{F}Q \to Q)$ such that $Uf = \mathit{id}$.
\end{defi}

\begin{lem}\label{lem:alg-iso}
  For each $F$-algebra $\alpha : FA \to A$, there is an isomorphism of
  categories $(\Alg_{\hat{F}})_\alpha \cong \Alg_{F^\alpha}$.
\end{lem}
\begin{proof}
  We demonstrate only the isomorphism on objects here; the isomorphism
  on morphisms is similar. An object of $(\Alg_{\hat{F}})_\alpha$ is a
  pair comprising a family $(A,P)$ and a morphism $k : \hat{F}(A,P)
  \to (A,P)$ in $\Fam(\Set)$ such that $Uk = \alpha$. By
  \lemref{lem:hidden-cartesian}, such morphisms $k$ are in one-to-one
  correspondence with the morphisms $k' : \hat{F}(A,P) \to
  \alpha^*(A,P)$ in $\Fam(\Set)_{FA}$. By the adjunction
  $\Sigma_\alpha \dashv \alpha^*$, the latter morphisms are in
  one-to-one correspondence with the morphisms $k'' : \Sigma_\alpha
  \hat{F}(A,P) \to (A,P)$ in $\Fam(\Set)_A$. By the definition of
  $F^\alpha$, these morphisms are exactly the $F^\alpha$-algebras,
  i.e., the objects of $\Alg_{F^\alpha}$.
\end{proof}

The next lemma shows that the reindexing and op-reindexing functors
for $\palg : \Alg_{\hat{F}} \to \Alg_F$ are inherited from $U :
\Fam(\Set) \to \Set$. We have:

\begin{lem}\label{lem:alg-bifibration}
  For every $F$-algebra morphism $f : (\alpha : FA \to A) \to (\beta :
  FB \to B)$, there are functors $f^{\reindAlg} :
  (\Alg_{\hat{F}})_\beta \to (\Alg_{\hat{F}})_\alpha$ and
  $\Sigma^{\mathsf{Alg}}_f : (\Alg_{\hat{F}})_\alpha \to
  (\Alg_{\hat{F}})_\beta$ such that $\Sigma^{\mathsf{Alg}}_f \dashv
  f^{\reindAlg}$. Moreover, for any $\hat{F}$-algebra $k :
  \hat{F}(A,P) \to (A,P)$, the $\hat{F}$-algebra
  $\Sigma^{\mathsf{Alg}}_f(k : \hat{F}(A,P) \to (A,P))$ has carrier
  $\Sigma_f(A,P)$, and for any $\hat{F}$-algebra $k' : \hat{F}(B,Q)
  \to (B,Q)$, the $\hat{F}$-algebra $f^{\reindAlg}(k' : \hat{F}(B,Q)
  \to (B,Q))$ has carrier $f^*(B,Q)$.
\end{lem}

\begin{proof}
  By \lemref{lem:alg-iso}, we can treat $(\Alg_{\hat{F}})_\alpha$ as
  if it were $\Alg_{F^\alpha}$, and $(\Alg_{\hat{F}})_\beta$ as if it
  were $\Alg_{F^\beta}$. In \autoref{sec:families}, we noted that for
  any $f : A \to B$, there are functors $f^* : \Fam(\Set)_B \to
  \Fam(\Set)_A$ and $\Sigma_f : \Fam(\Set)_A \to \Fam(\Set)_B$ such
  that $\Sigma_f \dashv f^*$. The lemma statement is now a consequence
  of \thmref{thm:alg-adjunctions}, provided we can establish the isomorphism
  $F^\beta \circ \Sigma_f \cong \Sigma_f \circ F^\alpha$. But we can
  verify the existence of such an isomorphism as follows: 
  \begin{center}
    \begin{tabular}{cll}
      & $\Sigma_f \circ F^\alpha$ & \\
      $=$     & $\Sigma_f \circ \Sigma_\alpha \circ \hat{F}_A$ & by the
      definition of $F^\alpha$ \\
      $\cong$ & $\Sigma_\beta \circ \Sigma_{Ff} \circ \hat{F}_A$ & since $f$ is an $F$-algebra morphism \\
      $\cong$ & $\Sigma_\beta \circ \hat{F}_B \circ \Sigma_f$ & by \lemref{lem:lifting-commute-sigma} \\
      $=$     & $F^\beta \circ \Sigma_f$ & by the definition of $F^\beta$
    \end{tabular}
  \end{center}
  This is exactly as required.
\end{proof}

We can now see that \lemref{lem:hidden-cartesian} generalises from the
categories in the families fibration to those in $\palg$. This gives:

\begin{lem}\label{lem:hidden-cartesian-alg}
  Let let $k_1 : \hat{F}(A,P) \to (A,P)$ and $k_2 : \hat{F}(B,Q) \to
  (B,Q)$ be objects of $(\Alg_{\hat{F}})_\alpha$ and
  $(\Alg_{\hat{F}})_\beta$, respectively, and let $f : (\alpha : FA
  \to A) \to (\beta : FB \to B)$ be an $F$-algebra morphism. The set
  of morphisms $h$ in $\Alg_{\hat{F}}$ from $k_1 : \hat{F}(A,P) \to
  (A,P)$ to $k_2 : \hat{F}(B,Q) \to (B,Q)$ such that $\palg h = f$ is
  isomorphic to the set of morphisms in $(\Alg_{\hat{F}})_\alpha$ from
  $k_1 : \hat{F}(A,P) \to (A,P)$ to $f^{\reindAlg}(k_2 : \hat{F}(B,Q)
  \to (B,Q))$.
\end{lem}
\begin{proof}
  The proof is tedious but not difficult. The key point entails
  constructing from each $\hat{F}$-algebra morphism $h : (A,P) \to
  (B,Q)$ such that $\palg h = f$ another $\hat{F}$-algebra morphism
  $h'' : (A,P) \to f^*(B,Q)$ such that $\palg h' = id$. This is made
  easier by observing that the definition of $f^{\reindAlg} :
  (\Alg_{\hat{F}})_{\beta} \to (\Alg_{\hat{F}})_{\alpha}$ obtained by
  applying \thmref{thm:alg-adjunctions} in the proof of
  \lemref{lem:alg-bifibration} is equivalent to the functor which on
  input $k : \hat{F}(B,Q) \to (B,Q)$ returns $\phi \circ (Ff)^*k
  \circ \hat{F}(f,\mathit{id})$, where $\phi : (Ff)^*\beta^*(B,Q) \to
  \alpha^* f^* (B,Q)$ is the isomorphism derived from the fact that
  $f$ is an $F$-algebra morphism.
\end{proof}

Putting this all together, we can now give our explicit
characterisation of $\mu F^\alpha$.

\begin{thm}\label{thm:muFalpha-initial}
  The functor $F^\alpha$ has an initial algebra with carrier
  $\Sigma_{\fold\alpha} \top(\mu F)$, i.e., with carrier $(A, \lambda
  a:A.\ \{ x : \mu F \sepbar \fold{\alpha} x = a \})$.
\end{thm}
\begin{proof}
  By \lemref{lem:alg-iso}, it suffices to show that the category
  $(\Alg_{\hat{F}})_\alpha$ has an initial object with carrier
  $\Sigma_{\fold{\alpha}} \top(\mu F)$. We construct an initial object
  in $(\Alg_{\hat{F}})_\alpha$ from the initial $\hat{F}$-algebra
  $\inn_{\hat{F}} : \hat{F}(\top (\mu F)) \to \top (\mu F)$ from
  \thmref{thm:initial-fhat-algebra}. Since $U^{\mathsf{Alg}}$ is a
  left adjoint, it preserves initial objects, so that
  $U^{\mathsf{Alg}}(\inn_{\hat{F}} : \hat{F}(\top (\mu F)) \to \top
  (\mu F))$ is the initial $F$-algebra $\inn_F : F(\mu F) \to \mu F$.
  We can apply $\Sigma^{\mathsf{Alg}}_{\fold{\alpha}}$ to the initial
  $\hat{F}$-algebra to get our candidate object
  $\Sigma^{\mathsf{Alg}}_{\fold\alpha}(\inn_{\hat{F}} : \hat{F}(\top
  (\mu F)) \to \top (\mu F))$. By \lemref{lem:alg-bifibration}, this
  candidate has carrier $\Sigma_{\fold\alpha}\top (\mu F)$, as
  required.

  To see that our candidate object is initial in
  $(\Alg_{\hat{F}})_\alpha$, let $k : \hat{F}(A,P) \to (A,P)$ be any
  object in $(\Alg_{\hat{F}})_\alpha$. Then

  \begin{tabular}{cl}
    &
    $(\Alg_{\hat{F}})_\alpha(\Sigma^{\mathsf{Alg}}_{\fold{\alpha}}(\inn_{\hat{F}}
    : \hat{F}(\top (\mu F)) \to \top (\mu F)), (k : \hat{F}(A,P) \to
    (A,P)))$ \\ 
    $\cong$& $(\Alg_{\hat{F}})_{\inn_{F}}((\inn_{\hat{F}} :
    \hat{F}(\top (\mu F)) \to 
    \top (\mu F)), \fold{\alpha}^{\reindAlg}(k : \hat{F}(A,P) \to
    (A,P)))$ \\ 
    & \hspace{11cm}by \lemref{lem:alg-bifibration} \\
    $\cong$& $\{ h : \Alg_{\hat{F}}((\inn_{\hat{F}} : \hat{F}(\top (\mu F)) \to
    \top (\mu F)), (k : \hat{F}(A,P) \to (A,P))) \sepbar \palg h =
    \fold\alpha\}$ \\ 
    & \hspace{11cm}by \lemref{lem:hidden-cartesian-alg} \\
  \end{tabular}

  \noindent
  Since $\inn_{\hat{F}} : \hat{F}(\top (\mu F)) \to \top (\mu F)$ is
  the initial ${\hat{F}}$-algebra and $\palg$ takes $\fold{k}$ to
  $\fold{\alpha}$, the final set in the above sequence has exactly one
  element. Thus there is exactly one morphism from
  $\Sigma^{\mathsf{Alg}}_{\fold{\alpha}}(\inn_{\hat{F}} : \hat{F}(\top
  (\mu F)) \to \top (\mu F))$ to $(k : \hat{F}(A,P) \to (A,P))$ in
  $(\Alg_{\hat{F}})_\alpha$, and so our candidate object is indeed
  initial in $(\Alg_{\hat{F}})_\alpha$.
\end{proof}

For readers familiar with fibred category theory, we briefly sketch
how our definitions and proofs may be generalised. We have been
careful to state the definition of $F^\alpha$ in terms of the abstract
structure we identified in \autoref{sec:families}. It can therefore be
generalised to any full cartesian Lawvere category with very strong
coproducts.  \hyperref[lem:alg-bifibration]{Lemmas
  \ref*{lem:alg-bifibration}} and \hyperref[lem:hidden-cartesian-alg]{
  \ref*{lem:hidden-cartesian-alg}}, as well as
\thmref{thm:muFalpha-initial}, can also be generalised. As was shown
by Hermida and Jacobs \cite{hermida98structural}, for any lifting
$\hat{F}$, the obvious generalisation of the functor $\palg :
\Alg_{\hat{F}} \to \Alg_F$ is a fibration. The generalisation of
\lemref{lem:alg-iso} is a result about the fibre categories of this
fibration, and the generalisation of \lemref{lem:alg-bifibration}
shows that it is a bifibration (i.e., that the re-indexing functors
have left adjoints). The generalisation of
\thmref{thm:muFalpha-initial} then follows from the Proposition 9.2.2
of Jacobs' book \cite{jacobs99book}, which relates initial objects in
the total category of a fibration with initial objects in the fibres.

\subsection{More Example Refinements}\label{sec:refining-examples}

The following explicit formulas can be used to compute refinements for
polynomial functors with respect to the families fibration:
\[\begin{array}{lll}
  Id^{\alpha}(A,P) &=& (A, \lambda a. \{ x :\{ (A,P)\} \sepbar \alpha
  (\pi_{(A,P)} x) = a \})\\
 &=& (A, \lambda a. \{ a' : A, p : Pa' \sepbar \alpha a' = a \}) \\ 
  K_B^{\alpha}(A,P) &=& (A, \lambda a. \{x : B \sepbar \alpha x = a \})\\
  (G + H)^\alpha(A, P)
  &=& (A, \lambda a.
  \begin{array}[t]{l}
    \!\!\{ x : G\ \{(A,P)\} \sepbar \alpha (\mathsf{inl} (G\pi_{(A,P)}
    x)) = a \} \\ 
    \!\! \mathrel+ \{x : H\ \{(A,P)\} \sepbar \alpha (\mathsf{inr} (H\pi_{(A,P)}
    x)) = a \})
  \end{array} \\
  &=& (A, \lambda a.\ G^{\alpha \circ \mathsf{inl}}Pa + H^{\alpha
    \circ \mathsf{inr}}Pa)\\ 
  (G \times H)^\alpha(A,P)
  &=& (A, \lambda a.\ \{
  \begin{array}[t]{l}
    x_1 : G\ \{(A,P)\},\ x_2 : H\ \{(A,P)\} \sepbar \\
    \quad \alpha (G\pi_{(A,P)} x_1, H\pi_{(A,P)} x_2) = a \})
  \end{array}
\end{array}\]
Refinements of the identity and constant functors are as expected.
Refinement splits coproducts of functors into two cases, 
specialising the refining algebra for each summand. It is not, however,
possible to decompose the refinement of a product of functors $G
\times H$ into refinements of $G$ and $H$, not even by algebras other
than $\alpha$. This is because $\alpha$ may need to relate multiple 
elements to the overall index. 

\begin{eorollary}\label{ex:refined-trees}
  We can refine $\mu F_{\tyname{Tree}}$ by the
  $F_{\tyname{Tree}}$-algebra $\mathit{sumAlg}$ given by
\begin{displaymath}
  \begin{array}{lcl}
    \mathit{sumAlg} & : & F_{\tyname{Tree}}\mathbb{Z} \to \mathbb{Z} \\
    \mathit{sumAlg}\ (\mathsf{Leaf}\ z) & = & z \\
    \mathit{sumAlg}\ (\mathsf{Node}\ (l,r)) & = & l + r
  \end{array}
\end{displaymath}
The fold of $\mathit{sumAlg}$ sums the values stored at the leaves of
a tree. It yields the refinement $\mu
F_{\tyname{Tree}}^{\mathit{sumAlg}}$ given by
\begin{eqnarray*}
  F_{\tyname{Tree}}^{\mathit{sumAlg}}(\mathbb{Z}, P) 
  = (\mathbb{Z}, \lambda n.\ \{z : \mathbb{Z} \sepbar z = n\} + \{ l, r
  : \mathbb{Z}, x_1 : P l, x_2 : P r \sepbar n = l + r\ \})
\end{eqnarray*}
By \thmref{thm:muFalpha-initial} and the definition of
$\Sigma_{\fold{\mathit{sumAlg}}}$ we have that the refinement $\mu
F_{\tyname{Tree}}^{\mathit{sumAlg}}$ is $\lambda n. \{ x : \mu
F_{\tyname{Tree}} \sepbar \fold{\mathit{sumAlg}}x = n \}$. This
refinement indexes the elements of $\mu F_{\tyname{Tree}}$ by the sums
of the values in their leaves. It corresponds to the Agda 2 declaration
\begin{verbatim}
data SumTree : Integer -> Set where
  SumLeaf : (z : Integer) -> SumTree z
  SumNode : (l r : Integer) -> SumTree l -> SumTree r -> SumTree (l + r)
\end{verbatim}
\end{eorollary}

\noindent
Note that in the second summand of
$F_{\tyname{Tree}}^{\mathit{sumAlg}}$ we have two recursive references
to $P$, each with a separate index, and that these indices are related
to the overall index $n$ as in the second case of
$\mathit{sumAlg}$. However, the basic refinement process developed in
this section cannot be used to require indices of subterms to be
related to one another in particular ways. For instance, it cannot
enforce the requirement that the two subtrees sum to the same value,
or that the tree satisfy some balance property. Indeed, if such
restrictions are imposed, then some elements of the underlying data
type may fail to be assigned an index. We show how to treat this via
partial assignment of indices in \autoref{sec:partial-refinement}.

\subsection{Limiting cases} 

The two limiting cases of refinement are deserving of attention.
Refining by the initial $F$-algebra $\inn_F : F(\mu F) \to \mu F$
gives a $\mu F$-indexed type inductively characterised as the least
fixed point of the functor $F^{\inn_F} = \Sigma_{\inn_F}
\hat{F}$. Since $\inn_F$ is an isomorphism, $\Sigma_{\inn_F}$ is as
well. Thus $F^{\inn_F} \cong \hat{F}$, so that $\mu F^{\inn_F} = \mu
\hat{F} = \top (\mu F)$. Taking, for each $x : \mu F$, the canonical
singleton set $1$ to be $\{x\}$, we can regard each element of $\mu F$
is its own index. By contrast, refinement by the final algebra $! : F1
\to 1$ gives a $1$-indexed type inductively characterised by
$F^!$. Since $F^! \cong F$, the inductive type $\mu F^!$ is actually
$\mu F$. Since $1$ is the canonical singleton set, all elements of
$\mu F$ have exactly the same index.  Refining by the initial
$F$-algebra thus has maximal discriminatory power, while refining by
the final $F$-algebra has no discriminatory power whatsoever.

\section{Starting with Already Indexed
  Types}\label{sec:indexed-refinement}  

The development in \autoref{sec:refining-inductive} assumes that the
type being refined is the initial algebra of an endofunctor $F$ on
$\Set$. This seems to preclude refining an inductive type that is
already indexed. But since we carefully identified the abstract
structure of $\Fam(\Set)$ needed to construct our refinements, our
results can be extended to {\em any} fibration having that
structure. We now show that, in particular, we can refine already
indexed types.

To this end, let $A$ be a set, and suppose we want to refine an
$A$-indexed type.  As we have seen, such types may be interpreted in
the category $\Fam(\Set)_A$. The carrier of an $F$-algebra $\alpha$
with respect to which we want to refine an already $A$-indexed type
will thus be an $A$-indexed set $B : A \to \Set$, and the resulting
refinement will be a type of the form $\forall a. B a \to \Set$, i.e.,
will be a family of sets that is doubly indexed by both $A$ and $B$.

Just as the categories of indexed sets comprise the category
$\Fam(\Set)$ in \autoref{sec:families}, the families indexed by
$A$-indexed sets comprise a category $\Fam(\Set)_A \times_{\Set}
\Fam(\Set)$. (Our notation is derived from the pullback construction
used to construct this category in the general setting; see below.)
Objects of $\Fam(\Set)_A \times_{\Set} \Fam(\Set)$ are pairs $(B, P)$,
where $B : A \to \Set$ and $P : \forall a. Ba \to \Set$, and morphisms
are pairs $(f,f^{\sim}) : (B,P) \to (C,Q)$, where $f : \forall a. B a
\to C a$ and $f^{\sim} : \forall a, b \in B a. P a b \to Q a (f a b)$.
And just as there is a functor $U : \Fam(\Set) \to \Set$ defined by
$U(A,P) = A$ on objects and $U(f,F^\sim) = f$ on morphisms, there is a
functor $U^A : \Fam(\Set)_A \times_{\Set} \Fam(\Set) \to \Fam(\Set)_A$
defined by $U^A(B,P) = B$ on objects and $U^A(f,f^{\sim}) = f$ on
morphisms. We may now recreate each of the structures we identified
for the families fibration in \autoref{sec:families} for the new
fibration given by $U^A$. We have:
\begin{iteMize}{$\bullet$}
\item {\em Fibres}: For each object $B$ of $\Fam(\Set)_A$, the fibre
  of $(\Fam(\Set)_A \times_{\Set} \Fam(\Set))$ over $B$ is the
  category $(\Fam(\Set)_A \times_{\Set} \Fam(\Set))_B$ consisting of
  objects of $\Fam(\Set)_A \times_{\Set} \Fam(\Set)$ whose first
  component is $B$, and morphisms $(f,f^{\sim})$, where $f =
  \mathit{id}$. By abuse of terminology, such morphisms are again said
  to be \emph{vertical}.
\item {\em Reindexing}: Given a morphism $f : B \to C$ in
  $\Fam(\Set)$, we can define the re-indexing functor $f^* :
  (\Fam(\Set)_A \times_{\Set} \Fam(\Set))_C \to (\Fam(\Set)_A
  \times_{\Set} \Fam(\Set))_B$ by composition, similarly to how
  reindexing is defined for the families fibration.
\item {\em Truth functor}: For each set $A$, we can define $\top^A :
  \Fam(\Set)_A \to \Fam(\Set)_A \times_{\Set} \Fam(\Set)$ by
  $\top^A(B) = (B, \lambda a\, b.\, 1)$. As in the families fibration,
  this mapping of objects to truth predicates extends to a functor,
  called the {\em truth functor} for $U^A$.
\item {\em Comprehension functor}: For each set $A$, we can define
  $\{-\}^A : \Fam(\Set)_A \times_{\Set} \Fam(\Set) \to \Fam(\Set)_A$
  by $\{(B,P)\}^A = \lambda a. \{ (b \in Ba, p \in P a b)\}$. As in
  the families fibration, this mapping of objects to their
  comprehensions extends to a functor, called the {\em comprehension
    functor} for $U^A$.
\item {\em Indexed coproducts}: For any morphism $f : B \to C$ in
  $\Fam(\Set)_A$, we can define $\Sigma_f : (\Fam(\Set)_A
  \times_{\Set} \Fam(\Set))_B \to (\Fam(\Set)_A \times_{\Set}
  \Fam(\Set))_C$ by
  \begin{displaymath}
    \Sigma_f (B,P) = (C,\, \lambda a\, c.\, \Sigma_{b \in Ba}.\, (c = f a
    b) \times P a b). 
  \end{displaymath}
\item {\em Indexed products}: For any morphism $f : B \to C$ in
  $\Fam(\Set)_A$, we can define $\Pi_f : (\Fam(\Set)_A \times_{\Set}
  \Fam(\Set))_B \to (\Fam(\Set)_A \times_{\Set} \Fam(\Set))_C$ by
  \begin{displaymath}
    \Pi_f (B,P) = (C,\, \lambda a\, c. \Pi_{b \in Ba}.\, (c = f a b) \to P a b)
  \end{displaymath}
\end{iteMize}
\noindent
Given these definitions, we can check by hand that they satisfy the
same relationships from \autoref{sec:families} that their counterparts
for the families fibration do. It is therefore possible to re-state
each of the definitions and results in
\hyperref[sec:fibrational-ind]{Sections \ref*{sec:fibrational-ind}}
and \hyperref[sec:refining-inductive]{ \ref*{sec:refining-inductive}}
for $U^A$, and, thereby, to derive refinements of already indexed
inductive types. The constructions that we carry out in the families
fibration in \hyperref[sec:partial-refinement]{Sections
  \ref*{sec:partial-refinement}} and \hyperref[sec:zygo-refine]{
  \ref*{sec:zygo-refine}} can similarly be carried out in $U^A$ as
well.

\vspace*{0.1in}

For readers familiar with fibred category theory, we now sketch how to
generalise the above construction to construct a suitable setting for
indexed refinement from any full cartesian Lawvere category with
products and very strong coproducts, provided these satisfy the
Beck-Chevalley condition for coproducts. For this we can use the
\emph{change-of-base} construction for generating new fibrations by
pullback~\cite{jacobs99book}. Indeed, if $A$ is an object of
$\mathcal E$, then the following pullback in $\mathrm{Cat}$, the large
category of categories and functors, constructs ${\mathcal E}_A
\times_{\mathcal B} {\mathcal E}$:
\begin{displaymath}
  \xymatrix{
    {{\mathcal E}_A \times_{\mathcal B} {\mathcal E}} \ar[r] \ar[d]_{U^A}
    \pullbackcorner
    &
    {\mathcal E} \ar[d]^U
    \\
    {{\mathcal E}_A} \ar[r]^{\{-\}}
    &
    {\mathcal B}
  }
\end{displaymath}
\noindent
Instantiating $\mathcal E$ to $\Fam(\Set)$ and $U$ to the families
fibration constructs $\Fam(\Set)_A \times_{\Set} \Fam(\Set)$ as
defined above, up to currying. Moreover, the following theorem shows
that all the structure we require for constructing refinements is
preserved by the change-of-base construction, and thus ensures that
the change-of-base construction can be iterated as often as desired.

\begin{thm}\label{thm:change-of-base}
  If $U$ is a full cartesian Lawvere category with products and with
  very strong coproducts satisfying the Beck-Chevalley condition for
  coproducts, then so is $U^A$.
\end{thm}

\begin{proof}\emph{(Sketch)}
  First, $U^A$ is well-known to be a fibration by its definition via
  the change-of-base construction \cite{jacobs99book}. The truth
  functor for $U^A$ is defined for objects $P$ in $\mathcal{E}_A$ by
  $\top^AP = (P, \top\{P\})$, and the comprehension functor for $U^A$
  is defined by $\{(P,Y)\}^A = \Sigma_{\pi_{P}}Y$, where $P \in
  \mathcal{E}_A$ and $Y \in \mathcal{E}_{\{P\}}$. Coproducts are
  defined directly using the coproducts of $U$.
\end{proof}

\begin{eorollary}\label{ex:indexed-refinement}
  To demonstrate the refinement of an inductive type which is already
  indexed we consider a small expression language of well-typed
  terms. Let $\mathcal{T} = \{ \mathsf{int}, \mathsf{bool} \}$ be the
  set of possible base types. The language is $\mu F_{\tyname{wtexp}}$
  for the functor $F_{\tyname{wtexp}} : \Fam(\Set)_{\mathcal{T}} \to
  \Fam(\Set)_{\mathcal{T}}$ given by
\begin{eqnarray*}
  F_{\tyname{wtexp}}(\mathcal{T},P) & = &
  (\mathcal{T}, \lambda t : \mathcal{T}.\begin{array}[t]{l}
    \{ z : \mathbb{Z} \sepbar t = \mathsf{int} \} \\
    +\ \{ b : \mathbb{B} \sepbar t = \mathsf{bool} \} \\
    +\ \{ x_1 :  P t,\ x_2 : P t \sepbar t = \mathsf{int} \} \\
    +\ \{ x_1 : P \mathsf{bool},\ x_2 : P t,\ x_3 : P t \})
  \end{array}
\end{eqnarray*}
This specification of an inductive type corresponds to the following
Agda 2 declaration, where we write \verb|Ty| for the Agda 2
equivalent of the set $\mathcal{T}$:
\begin{verbatim}
data WTExp : Ty -> Set where
  intConst  : Integer -> WTExp Int
  boolConst : Boolean -> WTExp Bool
  add       : WTExp Int -> WTExp Int -> WTExp Int
  if        : (t : Ty) -> WTExp Bool -> WTExp t -> WTExp t -> WTExp t
\end{verbatim}
The type \verb|WTExp| cannot be constructed by the process of
refinement presented in \autoref{sec:refining-inductive}. Indeed, the
indices of subexpressions, and not just the overall indexes, are
constrained in the types of the \texttt{add} and \texttt{if}
constructors. This accords with the discussion at the end of
\hyperref[sec:refining-examples]{Section
  \ref*{sec:refining-examples}}. Fortunately we can, and will, show in
\autoref{sec:partial-refinement} how to extend the notion of
refinement to the situation where not every element of a data type can
be assigned an index.

Meanwhile, in light of \thmref{thm:change-of-base}, we can refine the
already indexed type $\mu F_{\tyname{wtexp}}$. For any $t$, write
$\mathsf{IntConst}$, $\mathsf{BoolConst}$, $\mathsf{Add}$, and
$\mathsf{If}$ for the injections into $(\mathit{snd}\,
(F_{\tyname{wtexp}}(\mathcal{T},P)))\, t$. Let $\mathbb{B} = \{
\mathsf{true}, \mathsf{false} \}$ denote the set of booleans, and
assume there exists a $\mathcal{T}$-indexed family $T$ such that $T\
\mathsf{int} = \mathbb{Z}$ and $T\ \mathsf{bool} = \mathbb{B}$. Then
$T$ gives a semantic interpretation of the types from $\mathcal T$
that can be used to define an $F_{\tyname{wtexp}}$-algebra
$\mathit{evalAlg}$ whose fold specifies a ``tagless'' interpreter. We
have:
\begin{displaymath}
 \begin{array}{lll}
   \mathit{evalAlg} & : & F_{\tyname{wtexp}}(\mathcal{T},T) \to (\mathcal{T},T) \\
   \mathit{evalAlg} & = & (\mathit{id}, \lambda x : \mathcal{T}.\
   \lambda t : \mathit{snd}\, (F_{\tyname{wtexp}}
   (\mathcal{T},T))\, x.\
   \mathsf{case}\ t\ \mathsf{of} \\
   & & \quad\begin{array}{lll}
     \quad\mathsf{IntConst}\ z  & \Rightarrow & z\\
     \quad\mathsf{BoolConst}\ b  & \Rightarrow & b\\
     \quad\mathsf{Add}\ (z_1, z_2) &
           \Rightarrow  & z_1 + z_2 \\  
     \quad\mathsf{If}\ (b, x_1, x_2) &
           \Rightarrow & \mathsf{if}\  b\  \mathsf{then}\ 
           x_1\ \mathsf{else}\ x_2) 
       \end{array}
  \end{array}
\end{displaymath}
The function $\fold{\mathit{evalAlg}} : \forall t.\ \mu
F_{\tyname{wtexp}}t \to Tt$ does indeed give a semantics to each
well-typed expression. Refining $\mu F_{\tyname{wtexp}}$ by
$\mathit{evalAlg}$ yields an object \verb|WTExpSem| of
$\Fam(\Set)_{\mathcal{T}} \times_{\Set} \Fam(\Set)$ over
$(\mathcal{T},T)$, i.e, an object of $\Fam(\Set)$ indexed by
$\{(\mathcal{T},T)\}$. This $\{(\mathcal{T},T)\}$-indexed data type
associates to every well-typed expression that expression's
semantics. As an Agda 2 declaration, it can be expressed as follows,
after applying a few type isomorphisms to make the declaration more
idiomatic:
\begin{verbatim}
data WTExpSem : (t : Ty) -> T t -> Set where
  intConst  : (z : Integer) ->       WTExpSem Int z
  boolConst : (b : Boolean) ->       WTExpSem Bool b
  add       : (z1 z2 : Integer) ->
              WTExpSem Int z1 ->
              WTExpSem Int z2 ->     WTExpSem Int (z1 + z2)
  if        : (b : Boolean) ->
              (t : Ty) ->
              (x1 x2 : T t) ->
              WTExpSem Bool b ->
              WTExpSem t x1 ->
              WTExpSem t x2 ->       WTExpSem t (if b then x1 else x2)
\end{verbatim}
Here, we have assumed a standard \texttt{if\_then\_else} notation for
eliminating booleans.
\end{eorollary}

\section{Partial Refinement}\label{sec:partial-refinement}

In \hyperref[sec:refining-inductive]{Sections
  \ref*{sec:refining-inductive}} and
\hyperref[sec:indexed-refinement]{ \ref*{sec:indexed-refinement}} we
assumed that every element of an inductive type can be assigned an
index. Every list has a length, every tree has a number of leaves,
every well-typed expression has a semantic meaning, and so on. But how
can an inductive type be refined if only {\em some} data have values
by which we want to index?  For example, how can the inductive type of
well-typed expressions of \hyperref[ex:indexed-refinement]{Example
  \ref*{ex:indexed-refinement}} be obtained by refining a data type of
untyped expressions by an algebra for type assignment? And how can the
inductive type of red-black trees be obtained by refining a data type
of coloured trees by an algebra enforcing the well-colouring
properties? As these questions suggest, the problem of refining
subsets of inductive types is a common and naturally occurring
one. Our partial refinement technique, which we now describe, can
solve this problem.

\subsection{Partial Algebras}

To generalise our theory to partial refinements we move from algebras
to partial algebras. If $F$ is a functor, then a {\em partial
  $F$-algebra} is a pair $(A, \alpha : FA \to (1+A))$ comprising a
carrier $A$ and a structure map $\alpha : FA \to (1+A)$. We write
$\ok:A \to 1+A$ and $\fail:1 \to 1+A$ for the injections into $1+A$,
and often refer to a partial algebra solely by its structure map. The
functor $MA = 1+A$ is (the functor part of) the {\em error monad}.

\begin{eorollary} The inductive type of expressions is $\mu
  F_{\tyname{exp}}$ for the functor $F_{\tyname{exp}}X = \mathbb{Z} +
  \mathbb{B} + (X \times X) + (X \times X \times X)$.  Letting
  $\mathcal{T} = \{ \mathsf{int}, \mathsf{bool} \}$ as in
  \hyperref[ex:indexed-refinement]{Example
    \ref*{ex:indexed-refinement}}, and using the obvious convention
  for naming the injections into $F_{\tyname{exp}}X$, types can be
  inferred for expressions using the following partial
  $F_{\tyname{exp}}$-algebra:
\begin{displaymath}
  \begin{array}{lcl}
    \mathit{tyInfer} & : & F_{\tyname{exp}}\mathcal{T} \to 1 + \mathcal{T} \\
    \mathit{tyInfer}\ (\mathsf{IntConst}\ z) &=& \ok\ \mathsf{int} \\
    \mathit{tyInfer}\ (\mathsf{BoolConst}\ b) &=& \ok\ \mathsf{bool} \\
    \mathit{tyInfer}\ (\mathsf{Add}\ (t_1, t_2)) &=&
    \left\{
      \begin{array}{ll}
        \ok\ \mathsf{int} &
        \textrm{if }t_1 = \mathsf{int}\textrm{ and }t_2 = \mathsf{int} \\
        \mathsf{fail} & \textrm{otherwise}
      \end{array}
    \right. \\
    \mathit{tyInfer}\ (\mathsf{If}\ (t_1,t_2,t_3)) &=&
    \left\{
      \begin{array}{ll}
        \ok\ t_2 &
          \textrm{if }t_1 = \mathsf{ bool } \textrm{ and }t_2 = t_3\\
        \mathsf{fail} & \textrm{otherwise}
      \end{array}
    \right.
  \end{array}
\end{displaymath}
\end{eorollary}

\begin{eorollary}
  Let $\mathbb{C} = \{ \mathsf{R}, \mathsf{B} \}$ be a set of colours.
  The inductive type of coloured trees is $\mu F_{\tyname{ctree}}$ for
  the functor $F_{\tyname{ctree}}X = 1 + \mathbb{C} \times X \times
  X$.  We write $\mathsf{Leaf}$ and $\mathsf{Br}$ for injections into
  $F_\tyname{ctree}X$. Red-black trees~\cite{cormen01intro} are
  coloured trees satisfying the following constraints:
\begin{enumerate}
\item Every leaf is black;
\item Both children of a red node are black;
\item For every node, all paths to leaves contain the same number of
  black nodes.
\end{enumerate}
We can check whether or not a coloured tree is a red-black tree using
the following partial $F_{\tyname{ctree}}$-algebra. Its carrier
$\mathbb{C} \times \mathbb{N}$ records the colour of the root in the
first component and the number of black nodes to any leaf, assuming
this number is the same for every leaf, in the second. We have:
\begin{displaymath}
  \begin{array}{lll}
    \mathit{checkRB}  &:& F_{\tyname{ctree}}(\mathbb{C} \times
    \mathbb{N}) \to 1 + (\mathbb{C} \times \mathbb{N}) \\ 
    \mathit{checkRB}\ \mathsf{Leaf} & = & 
    \ok\ (\mathsf{B}, 1) \\
    \mathit{checkRB}\ (\mathsf{Br}\ (\mathsf{R}, (s_1,n_1), (s_2, n_2))) & = & 
    \left\{
      \begin{array}{ll}
        \ok\ (\mathsf{R}, n_1) &
          \textrm{if }s_1 = s_2 = \mathsf{B} 
          \textrm{ and }n_1 = n_2 \\
        \mathsf{fail}        & \textrm{otherwise}
      \end{array}
    \right.\\
    \mathit{checkRB}\ (\mathsf{Br}\ (\mathsf{B}, (s_1,n_1), (s_2, n_2))) & = & 
    \left\{
      \begin{array}{ll}
        \ok\ (\mathsf{B}, n_1 + 1) & \textrm{if }n_1 = n_2 \\
        \mathsf{fail}        & \textrm{otherwise}
      \end{array}
    \right.
  \end{array}
\end{displaymath}
\end{eorollary}

\subsection{Using a Partial Algebra to Select Elements}

We now show how, given a partial algebra, we can use it to select some
of the elements of an underlying type and assign them indices. The key
to doing this is to turn every partial $F$-algebra into a (total)
$F$-algebra. Let $\lambda : F \circ M \to M \circ F$ be any
distributive law for the error monad $M$ over the functor $F$. Then
$\lambda$ respects the unit and multiplication of $M$
(see~\cite{barr83toposes} for details). Every partial $F$-algebra
$\kappa:FA \to (1+A)$ generates an $F$-algebra $\overline{\kappa} :
F(1 +A) \to (1 + A)$ defined by $\overline{\kappa} = [\fail, \kappa]
\circ \lambda_A$, where $[\fail, \kappa]$ is the cotuple of the
functions $\fail$ and $\kappa$.

We can use $\overline{\kappa}$ to construct the following global
characterisation of the indexed type for which we seek an inductive
characterisation:
\begin{displaymath}
  (A, \lambda a.\ \{ x : \mu F \sepbar \fold{\overline{\kappa}} x =
  \ok\ a \}) 
\end{displaymath}
As in~\parenref{eqn:vector}, we can consider this characterisation a
specification; it is similar to the specification
in~\hyperref[sec:refining-inductive]{Section
  \ref*{sec:refining-inductive}}, except that the index generated by
the algebra $\overline{\kappa}$ is required to return $\ok\, a$ for
some $a \in A$. We can rewrite this specification as follows, using
the categorical constructions from \autoref{sec:families} and
\thmref{thm:muFalpha-initial}:
\begin{equation}
  \label{eqn:muFkappa}  
  (A, \lambda a.\ \{ x : \mu F \sepbar \fold{\overline{\kappa}} x =
  \ok\, a \}) = \ok^* \circ \Sigma_{\fold{\overline{\kappa}}} \top (\mu F) =
  \ok^* \mu F^{\overline{\kappa}}
\end{equation}
Rewriting the specification in this way links partial refinements with
the indexed inductive type generated by the refinement process given
in \autoref{sec:refining-inductive}.

\subsection{Construction and Correctness of Partial Refinement}

Refining $\mu F$ by the $F$-algebra $\overline{\kappa}$ using the
techniques of \autoref{sec:refining-inductive} would result in an
inductive type indexed by $1+A$. But our motivating examples suggest
that what we actually want is an $A$-indexed type that inductively
describes only those terms having values of the form $\ok\ a$ for some
$a \in A$. Partial refinement constructs, from a functor $F$ with
initial algebra $\inn_F : F (\mu F) \to \mu F$, and a partial
$F$-algebra $\kappa : FA \to 1+A$, a functor $F^{?\kappa}$ such that
$\mu F^{?\kappa} \cong (A, \lambda a.\ \{ x : \mu F \sepbar
\fold{\overline{\kappa}} x = \ok\ a \}) = \ok^* \mu
F^{\overline{\kappa}}$. To this end, we define
\begin{equation}\label{eqn:Fkappa}
  F^{?\kappa} = \ok^* \circ \Sigma_{\kappa} \circ \hat{F}_A
\end{equation}
We note that, in the special case of the families fibration, this
definition specialises to $F^{?\kappa} = (A, \lambda a. \{ x :
F\{(A,P)\} \sepbar \kappa(F\pi_{(A,P)}x) = \ok\ a \})$. Now, since left
adjoints preserve initial objects, we can prove $\mu F^{?\kappa} \cong
\ok^* \mu F^{\overline{\kappa}}$ by lifting the adjunction on the left
below (cf.~\autoref{sec:indexed-products}) to an adjunction between
$\Alg_{F^{?\kappa}}$ and $\Alg_{F^{\overline{\kappa}}}$ via
\thmref{thm:alg-adjunctions}:
\begin{displaymath}
\begin{array}{lcl}
  \adjunction{\Fam(\Set)_A}{\Fam(\Set)_{1 +
      A}}{\ok^*}{\Pi_{\ok}} 
& \quad\Rightarrow\quad
&
  \adjunction{\Alg_{F^{?\kappa}}}{\Alg_{F^{\overline{\kappa}}}}{}{}  
\end{array}
\end{displaymath} 

\noindent
To satisfy the precondition of \thmref{thm:alg-adjunctions}, we must
prove that $F^{?\kappa} \circ \ok^* \cong \ok^* \circ
F^{\overline{\kappa}}$. To show this, we reason as follows:
\begin{center}
  \begin{tabular}{cll}
    & $\ok^* \circ F^{\overline{\kappa}}$ & \\
    $=$     & $\ok^* \circ \Sigma_{\overline{\kappa}} \circ \hat{F}_A$ &
    by definition of $F^{\overline{\kappa}}$\\ 
    $\cong$ & $\ok^* \circ \Sigma_{\kappa} \circ (F\,\ok)^* \circ
    \hat{F}_A$ & by \lemref{lem:commute-ok} below \\ 
    $\cong$ & $\ok^* \circ \Sigma_{\kappa} \circ \hat{F}_{1+A} \circ \ok^*$ & by
    \lemref{lem:lifting-commute-reindex} \\ 
    $=$     & $F^{?\kappa} \circ \ok^*$ & by definition of $F^{?\kappa}$
  \end{tabular}
\end{center}
In these steps we have made use of two auxiliary results, relying on
two assumptions. First, in order to apply
\lemref{lem:lifting-commute-reindex}, we have assumed that $F$
preserves pullbacks. Secondly, we have made use of the vertical
natural isomorphism $\ok^* \circ \Sigma_{\overline{\kappa}} \cong
\ok^* \circ \Sigma_{\kappa} \circ (F \ok)^*$. We may deduce the
existence of the latter if we assume that the following property,
which we call \emph{non-introduction of failure}, is satisfied by the
distributive law $\lambda$ for the error monad $M$ over $F$: for all
$x : F(1 + A)$ and $y : FA$, $\lambda_A\, x = \ok\ y$ if and only if
$x = F\, \ok\, y$.  This property strengthens the usual unit axiom for
distributive laws in which the implication holds only from right to
left, and ensures that if applying $\lambda$ does not result in
failure, then no failures were present in the data to which $\lambda$
was applied. Every container functor has a canonical distributive law
for $M$ satisfying the non-introduction of failure property.

\begin{lem}\label{lem:commute-ok}
  If the distributive law $\lambda$ satisfies non-introduction of
  failure, then $\ok^* \circ \Sigma_{\overline{\kappa}} \cong \ok^*
  \circ \Sigma_{\kappa} \circ (F \ok)^*$.
\end{lem}

\begin{proof}
  Given $(F(1+A), P : F(1 + A) \to \Set)$, we have    
  \begin{eqnarray*}
    & & (\ok^* \circ \Sigma_{\overline{\kappa}})(F(1+A),P)\\
    & = & (A, \lambda a:A.\ \{ (x_1 : F(1 + A), x_2 : Px_1) \sepbar [\fail,
    \kappa](\lambda_A x_1) = \ok\ a \}) \\ 
    &\cong& (A, \lambda a:A.\ \{ x_1 : FA, x_2 : P(F\, \ok\, x_1) \sepbar
    \kappa x_1 = \ok\ a \}) \\ 
    & \cong & (A, \ok^* \circ \Sigma_{\kappa} \circ (F\,
    \ok)^* (F(1+A),P))    
  \end{eqnarray*}
  Here, we have instantiated the definitions in terms of the
  constructions from \autoref{sec:families} for the families
  fibration.
\end{proof}

\noindent
Putting everything together, we have shown the correctness of partial
refinement:

\begin{thm}\label{thm:partial-refinement}
  If $\lambda$ is a distributive law for the error monad $M$ over $F$
  with the non-introduction of failure property, and if $F$ preserves
  pullbacks, then $F^{?\kappa}$ has an initial algebra whose carrier
  is given by any, and hence all, of the expressions
  in~\parenref{eqn:muFkappa}.
\end{thm}

\noindent
In fact, Lemma 6.1, and hence Theorem 6.2, holds in the more general
setting of a full cartesian Lawvere category with products and very
strong coproducts that satisfy the Beck-Chevalley condition for
coproducts, provided that the base category satisfies
extensivity~\cite{clw93}. In the general setting, the non-introduction
of failure property can be formulated as requiring that the following
square (which is the unit axiom for the distributive law $\lambda$) is
a pullback:
\begin{displaymath}
  \xymatrix{
    {FA} \ar[r]^(.4){F\ok} \ar[d]_{\mathit{id}}
    &
    {F(1+A)} \ar[d]^{\lambda_A}
    \\
    {FA} \ar[r]^(.4)\ok
    &
    {1 + FA}
  }
\end{displaymath}
Moreover,~\thmref{thm:change-of-base} extends to show that extensivity
is also preserved by change-of-base provided all of the the fibres of
the given full cartesian Lawvere category satisfy extensivity. This
ensures that the process of partial refinement can be iterated as
often as desired.

\section{Refinement by Zygomorphisms and Small Indexed
Induction-Recursion}\label{sec:zygo-refine}

The refinement process of \autoref{sec:refining-inductive} allows us
to refine an inductive data type by any function definable as a fold.
Despite this generality, the restriction to functions defined by folds
can be a burden. Consider, for example, the following structurally
recursive function on natural numbers that computes factorials:
\begin{verbatim}
factorial : Nat -> Nat
factorial zero     = succ zero
factorial (succ n) = succ n * factorial n
\end{verbatim}
This \verb|factorial| function is not immediately expressible as a
fold of an algebra on the natural numbers; indeed, the right-hand side
of the second clause uses both the result of a recursive call and the
current argument, but a fold cannot use the current argument in
computing its result. The style of definition exemplified by
\verb|factorial| is known as a
\emph{paramorphism}~\cite{meertens92paramorphism}. As we recall in
\autoref{sec:para-zygo} below, such definitions can be reduced to
folds. However, reducing \verb|factorial| to a fold and then refining
as in \autoref{sec:refining-inductive} yields a $($\texttt{Nat
  $\times$ Nat}$)$-indexed type, i.e., a doubly indexed type that
reveals the auxiliary data used to define~\texttt{factorial} as a
fold.  But rather than $($\texttt{Nat $\times$ Nat}$)$-indexed type,
what we actually want is an inductive characterisation of the
following \texttt{Nat}-indexed type:
\begin{equation}\label{eq:suffixlist}
  \texttt{FactorialNat n} \cong 
\{ \texttt{x} : \texttt{Nat} \sepbar \texttt{factorial x} = \texttt{n} \}
\end{equation}
If we try to implement \texttt{FactorialNat} inductively in Agda 2,
then we get stuck at the point marked by \texttt{???} below:
\begin{verbatim}
data FactorialNat : Nat -> Set where
  fnzero : FactorialNat (succ zero)
  fnsucc : {n : Nat} ->
           (x : FactorialNat n) ->
           FactorialNat (succ ??? * n)
\end{verbatim}
We'd like to put \verb|x| in place of \verb|???|, but there is a
problem. Indeed, if \verb|x : FactorialNat n|, then
in~\parenref{eq:suffixlist} we know that \verb|x : Nat|, so we can use
the assertion \verb|factorial x = n|.  But in the above Agda 2 code we
cannot conclude that if \verb|x : FactorialNat n|, then 
\verb|x : Nat|, and so we cannot use the fact that 
\verb|factorial x = n|. What is required is a function
{\tt forget} of type {\tt {n : Nat} -> FactorialNat n -> Nat} that
converts an element of \verb|FactorialNat n| into its underlying
natural number. Unfortunately, we cannot first define the data type
\verb|FactorialNat| and then define the function \verb|forget|
thereafter. Instead, as becomes evident upon replacing \verb|???| by
\verb|forget x| in the definition of \verb|FactorialNat|, we must
define both simultaneously.

Fortunately, this can be done using the principle of definition by
\emph{indexed induction-recursion} (IIR) due to Dybjer and Setzer
\cite{dybjer03induction,dybjer06indexed}. Agda 2 supports indexed
induction recursion, and so \texttt{FactorialNat} and \texttt{forget}
can be defined (simultaneously) as follows:
\begin{verbatim}
mutual
  data FactorialNat : Nat -> Set where
    fnzero : FactorialNat (succ zero)
    fnsucc : {n : Nat} ->
             (x : FactorialNat n) ->
             FactorialNat (succ (forget x) * n)

  forget : {n : Nat} -> FactorialNat n -> Nat
  forget fnzero     = zero
  forget (fnsucc x) = succ (forget x)
\end{verbatim}
\noindent
As we have already noted, it is possible to make sense of functions
such as \texttt{factorial} in terms of initial $F$-algebras by using the
existing notion of a \emph{paramorphism} and its generalisation, a
\emph{zygomorphism}, but this gives incorrectly indexed types.
Instead, making use of a presentation of inductive-recursive
definitions as initial algebras (\autoref{sec:initial-algebra-ir}), we
show in \autoref{sec:zygo-refine-correct} that the definition of
\texttt{FactorialNat} can be generalised to an inductive-recursive
type satisfying the analogue of \parenref{eq:suffixlist} for all
zygomorphisms (rather than just \verb|factorial|) and all initial
algebras of functors (rather than just \verb|Nat|).

\subsection{Zygomorphisms and Paramorphisms}\label{sec:para-zygo}

\emph{Zygomorphisms} were introduced by Malcolm \cite{malcolm90}, and
have as a special case the concept of a \emph{paramorphism}
\cite{meertens92paramorphism}. Given a morphism $\gamma : F(D \times
A) \to A$ and an $F$-algebra $\delta : FD \to D$ we define the
$F$-algebra $\overline{\gamma,\delta} : F(D \times A) \to D \times A$
by $\langle \delta \circ F\pi_1, \gamma\rangle$.
%\begin{displaymath}
%  \xymatrix{
%    {F(D \times A)} \ar[r]^(.4){\langle F\pi_1, \mathit{id}\rangle}
%    &
%    {FD \times F(D \times A)} \ar[r]^(.65){\delta \times \gamma}
%    &
%    {D \times A}
%  }
%\end{displaymath}
The zygomorphism $h$ associated with $\overline{\gamma,\delta}$ is
defined to be $\pi_2 \circ \fold{\overline{\gamma,\delta}} : \mu F \to
A$. It is the unique morphism satisfying the equation $h \circ \inn_F
= \gamma \circ F\langle\fold{\delta},h\rangle$. Paramorphisms are a
special case of zygomorphisms for which $\delta$ is the initial
$F$-algebra $\inn_F : F(\mu F) \to \mu F$.

The \texttt{factorial} function above can be represented as a
paramorphism (and hence as a zygomorphism). Recalling that the carrier
of the initial algebra for the functor $F_{\tyname{Nat}}X = 1 + X$ is
$\mathbb{N}$, we can define
\begin{equation}\label{eq:fact-paramorphism}
  \begin{array}{lll}
    \mathit{fact} & : & F_{\tyname{Nat}}(\mathbb{N} \times \mathbb{N}) \to
\mathbb{N} \\
    \mathit{fact}\ \mathsf{zero} & = & 1 \\
    \mathit{fact}\ (\mathsf{succ}\ (n,x)) & = & (n + 1) * x
  \end{array}
\end{equation}
Here, we have used $\mathsf{zero}$ and $\mathsf{succ}$ as suggestive
names for the two injections into $1 + X$. Taking $\gamma$ to be
$\mathit{fact}$, the induced paramorphism from $\mathbb{N}$ to
$\mathbb{N}$ is exactly the factorial function.

\subsection{Initial Algebra Semantics of Indexed Small
Induction-Recursion}\label{sec:initial-algebra-ir}

Indexed induction-recursion allows us to define a family of types $X :
A \to \Set$ simultaneously with a recursive function $f : \forall a.\
Xa \to Da$, for some $A$-indexed collection of potentially large types
$Da$. We are interested in the case when $D$ does not depend on $A$,
so that $Da$ is $D$, and $D$ is small, i.e., $D$ is a set. In this
situation, the semantics of IIR definitions can be given as initial
algebras of functors over slice categories. We recall the definition
of slice categories on $\Set$. Given a set $D$, the {\em slice
  category} $\Set/D$ on $\Set$ has as objects pairs $(Z : \Set, f : Z
\to D)$. A morphism from $(Z,f)$ to $(Y,g)$ in $\Set/D$ is a function
from $h : Z \to Y$ such that $f = g \circ h$. We write $f$ for $(Z,f)$
when $Z$ can be inferred from context.

Noting that $\forall a.Xa \to D$ is isomorphic to $(\Sigma a. Xa) \to
D$ and that $\Sigma a. Xa = \{(A,X)\}$, this leads us to consider the
category $\Set^A \times_\Set \Set/D$ each of whose objects is an
$A$-indexed set $X$ together with a function from $\{(A,X)\}$ to
$D$. A morphism in this category from $(X, f)$ to $(X', g)$ is a
function $\phi : \forall a. Xa \to X'a$ such that $\forall a : A.\, p:
Xa.\, f(a,p) = g (a, \phi a p)$. In fact, this category is the
following pullback:
\begin{displaymath}
  \xymatrix{
    {\Set^A \times_\Set \Set/D} \ar[r] \ar[d] \pullbackcorner
    &
    {\Set/D} \ar[d]^{\pi_1}
    \\
    {\Set^A} \ar[r]^{\{-\}}
    &
    {\Set}
  }
\end{displaymath}

The pair (\texttt{FactorialNat}, \texttt{forget}) can be interpreted
as the carrier of the initial algebra of the following functor on
$\Set^{\mathbb{N}} \times_{\Set} \Set/\mathbb{N}$:
\begin{equation}\label{eq:factorialnat}
  \begin{array}[t]{l}
      F_{\tyname{FactorialNat}}(X : \Set^{\mathbb{N}}, f : \{(\mathbb{N},X)\}
\to \mathbb{N}) = \\
      \quad
      \begin{array}{l}
        (\lambda n.\ \{* \sepbar n = 1\} + \{(n_1 : \mathbb{N}, x : Xn_1)
\sepbar n = (n_1 + 1) * f (n_1, x) \}, \\
        ~\lambda(n,x).
        \begin{array}[t]{l}
          \mathsf{case}\ x\ \mathsf{of}\\
          \quad\mathsf{inl}\ * \Rightarrow 0 \\
          \quad\mathsf{inr}\ (n_1,x) \Rightarrow f(n_1,x) + 1)
        \end{array}
      \end{array}
  \end{array}
\end{equation}
The first component of $F_{\tyname{FactorialNat}}(X,f)$ defines the
constructors of \texttt{FactorialNat} in a manner similar to that
described in \hyperref[sec:indexed-ind-types]{Section
  \ref*{sec:indexed-ind-types}}. Note that this first component
depends on both $X$ and $f$, which is characteristic of
inductive-recursive, as well as of indexed inductive-recursive,
definitions. The second component of $F_{\tyname{FactorialNat}}(X,f)$
extends the function $f$ to the new cases given in the first component
of $F_{\tyname{FactorialNat}}(X,f)$.

To develop refinement by zygomorphisms, we use a similar methodology
to that in \autoref{sec:partial-refinement}. We first use the
refinement process of \autoref{sec:refining-inductive} to generate a
functor on $\Set^{D \times A}$ which has an initial algebra, and then
apply \thmref{thm:alg-adjunctions} with the adjoint equivalence in the
next theorem to produce the initial algebra for the functor on $\Set^A
\times_\Set \Set/D$ that we define
in \parenref{eq:zygo-refine-functor} below.

\begin{thm}\label{thm:ir-equiv}
  There is an adjoint equivalence $\Set^A \times_\Set \Set/D \simeq
  \Set^{D \times A}$ which is witnessed by the following pair of
  functors:
\[\begin{array}{lll}
  \Psi & : & \Set^{D \times A} \to \Set^A \times_{\Set} \Set/D \\
  \Psi(X) & = & (\,\lambda a.\, \{ (d, x) \sepbar d : D, x : X (d,a) \},\; 
  \lambda (a, (d,x)).d)\\
  & & \\
  \Phi & : & \Set^A \times_{\Set} \Set/D \to \Set^{D \times A} \\
  \Phi(X,f) &=& \lambda(d,a).\, \{ x : X a \sepbar f (a,x) = d \}
  \end{array}\]
\end{thm}
\begin{proof}
  This is a simple consequence of the fact that, for any set $X$,
 $\Set^X \simeq \Set/X$.
\end{proof}
\noindent
In light of the equivalence demonstrated in \thmref{thm:ir-equiv}, we
could use $\Set^{D \times A}$, rather than $\Set^A \times_\Set
\Set/D$, as the appropriate category for refinement by
zygomorphisms. Our reasons for choosing the latter are twofold. First,
as we noted in the introduction to this section, we want an
$A$-indexed type rather than a $(D\times A)$-indexed type. Secondly,
we want to define a function from that $A$-indexed type into $D$
itself, rather than into a $D$-indexed type.

\subsection{Refinement by Zygomorphisms}\label{sec:zygo-refine-correct}

We now show how to refine an inductive type by a zygomorphism to
obtain an indexed inductive-recursive definition. Generalising the
example of \texttt{FactorialNat} above, we want to construct from an
$F$-algebra $\delta : FD \to D$ and a morphism $\gamma : F(D \times A)
\to A$ an inductive-recursive characterisation of the following
$A$-indexed set and accompanying $D$-valued function:
\begin{equation}\label{eq:zygo-refine-spec}
  (\lambda a.\, \{ (d : D, x : \mu F) \sepbar
  \fold{\overline{\gamma,\delta}} x = (d,a) \},\, \lambda (a,(d,x)).\
  d) : \Set^A \times_{\Set} \Set/D
\end{equation}
Note that although the fold $\fold{\overline{\gamma,\delta}}$
applied to $x$ produces a pair $(d,a)$, the first component of the
pair in \parenref{eq:zygo-refine-spec} is an $A$-indexed set, rather
than an $(A \times D)$-indexed set. We can now see that the object of
$\Set^A \times_\Set \Set/D$ in \parenref{eq:zygo-refine-spec} is
isomorphic to 
\begin{equation}\label{eq:zygo-refine-equiv}
  (\lambda a.\, \{ x : \mu F \sepbar
  \pi_2(\fold{\overline{\gamma,\delta}} x) = a \}, \lambda (a,
  x). \fold{\delta}x) 
\end{equation}
The first component of \parenref{eq:zygo-refine-equiv}, and hence the
first component of \parenref{eq:zygo-refine-spec}, is the refinement
of $\mu F$ by the zygomorphism $\pi_2 \circ
\fold{\overline{\gamma,\delta}}$, and is thus is the $A$-indexed set
we want to characterise inductively. To do this, we
characterise \parenref{eq:zygo-refine-spec} inductively. More
specifically, we prove in \thmref{thm:zygo-refine-correct} below that
the least fixed point of the following functor on $\Set^A \times_\Set
\Set/D$ gives an inductive-recursive characterisation of
\parenref{eq:zygo-refine-spec}:
\begin{equation}\label{eq:zygo-refine-functor}
  F^{\gamma,\delta}(X,f) =
  \begin{array}{l}
    (\lambda a.\ \{ x : F\{(D \times A, \Phi(X,f))\} \sepbar \gamma(F\pi_{(D \times A, \Phi(X,f))} x) = a \}, \\
    \quad\lambda (a,x).\ \delta(F\pi_1(F\pi_{(D \times A, \Phi(X,f))} x)))
  \end{array}
\end{equation}
\noindent
This definition makes use of the functor $\Phi : \Set^A \times_\Set
\Set_{/D} \to \Set^{A \times D}$ defined in \thmref{thm:ir-equiv}. The
first component of $F^{\gamma,\delta}(X,f)$ uses $\Phi$ to bundle up
$X$ and $f$ into a $(D \times A)$-indexed set, and then applies
$\Sigma_\gamma \circ \hat{F}$ as in the basic refinement construction
in \autoref{sec:refining-inductive}. The second component of
$F^{\gamma,\delta}(X,f)$ extracts the underlying $FD$ component of $x$
and then applies $\delta$.

\begin{eorollary}\label{ex:factorial-refine}
  We instantiate the characterisation of $F^{\gamma,\delta}$
  in \parenref{eq:zygo-refine-functor} for the \texttt{factorial}
  function from the introduction to this section. That is, we consider
  the functor $F_{\tyname{Nat}}X = 1 + X$, the $F$-algebra
  $\inn_{F_{\tyname{Nat}}} : F_{\tyname{Nat}} \mathbb{N} \to
  \mathbb{N}$, and the morphism $\mathit{fact} : F_{\tyname{Nat}}
  (\mathbb{N} \times \mathbb{N}) \to \mathbb{N}$ defined
  in \parenref{eq:fact-paramorphism}.  Instantiating
  \parenref{eq:zygo-refine-functor} gives
  \begin{displaymath}
    \begin{array}{cl}
      & F_{\tyname{Nat}}^{\mathit{fact},\inn_{F_{\tyname{Nat}}}}(X,f) \\
      =&
      \begin{array}{l}
        (\lambda n.\{ x : F_{\tyname{Nat}}\{(D \times A, \Phi(X,f)\}) \sepbar \mathit{fact}(F_{\tyname{Nat}}\pi_{(D \times A, \Phi(X,f))}x) = n \}, \\
        \quad\lambda (n,x).\inn_{F_{\tyname{Nat}}} (F_{\tyname{Nat}}\pi_1(F_{\tyname{Nat}}\pi_{(D \times A, \Phi(X,f))}x)))
      \end{array}
 \\
      =&
      \begin{array}{l}
        (\lambda n. \{ x : 1 + \{(D \times A, \Phi(X,f))\} \sepbar \mathit{fact}((1 + \pi_{(D \times A, \Phi(X,f))})x) = n \},\\
        \quad\lambda (n,x). \inn_{F_{\tyname{Nat}}} ((1 + \pi_1)((1 + \pi_{(D \times A, \Phi(X,f))})x)))
      \end{array}
    \end{array}
  \end{displaymath}
  We can rewrite the first component of
  $F^{\mathit{fact},\inn_{F_{\tyname{Nat}}}}_{\tyname{Nat}}(X,f)$ to the following
  $\mathbb{N}$-indexed set depending on $X$ and $f$:
  \begin{displaymath}
    \lambda n. \{ * \sepbar \mathit{fact}(\mathsf{zero}) = n \} + \{ (d,n_1),
    x : X n_1 \sepbar f(n_1,x) = d, \mathit{fact}(\mathsf{succ}(d,n_1)) = n \}.
  \end{displaymath}
  The $d$ component in the second summand above is constrained to be
  $f(n_1,x)$, so we can first remove all references to $d$ and then
  rewrite according to the definition of $\mathit{fact}$ to obtain
  \begin{displaymath}
    \lambda n. \{ * \sepbar 1 = n \} + \{ n_1, x : X n_1 \sepbar 
    (f(n_1,x) + 1) * n_1 = n \}
  \end{displaymath}
  Using this rewriting of the first component of the instantiation, we
  can rewrite the second component of
  $F^{\mathit{fact},\inn_{F_{\tyname{Nat}}}}_{\tyname{Nat}}(X,f)$ to
  use pattern matching and normal arithmetic notation to get
  \begin{displaymath}
    \lambda (n,x).\,
    \mathsf{case}\ x\ \mathsf{of}
    \left\{
      \begin{array}{lcl}
        \mathsf{zero} &\Rightarrow& 0 \\
        \mathsf{succ}(n_1,x) &\Rightarrow& f(n_1,x) + 1
      \end{array}
    \right.
  \end{displaymath}
  We have thus derived the definition of $F_{\tyname{FactorialNat}}$
  from \parenref{eq:factorialnat} solely by way of a mechanical
  process, using the components of the paramorphism that computes
  factorials. Moreover, by \thmref{thm:zygo-refine-correct} below, we
  know that this functor has an initial algebra, and that this initial
  algebra represents the refinement of the natural numbers by the
  zygomorphism defining the function \verb|factorial|.
\end{eorollary}

As described above, the correctness of refinement by a zygomorphism is
a consequence of \thmref{thm:alg-adjunctions} and the adjoint
equivalence from \thmref{thm:ir-equiv}. Indeed, we have:

\begin{thm}\label{thm:zygo-refine-correct}
  The functor $F^{\gamma,\delta} : \Set^A \times_\Set \Set/D \to
  \Set^A \times_\Set \Set/D$ defined
  in \parenref{eq:zygo-refine-functor} has an initial algebra whose
  carrier is given in \parenref{eq:zygo-refine-spec}.
\end{thm}

\begin{proof}
  Observe that the object of $\Set^A \times_{\Set} \Set/D$
  in \parenref{eq:zygo-refine-spec} is isomorphic to the result of
  applying the functor $\Psi$ defined in \thmref{thm:ir-equiv} to the
  result of refining $\mu F$ by the algebra
  $(\overline{\gamma,\delta}) : F(D \times A) \to D \times
  A$. Indeed,
  \begin{eqnarray*}
    &     & \Psi (\mu F^{\overline{\gamma,\delta}}) \\
    &\cong& \Psi (\lambda (d,a).\, \{ x : \mu F \sepbar
    \fold{\overline{\gamma,\delta}}x = (d,a) \}) \\
    &=    & (\lambda a.\, \{ (d : D, x : \mu F) \sepbar
    \fold{\overline{\gamma,\delta}}x = (d,a) \},\; \lambda (a,(d,x)).d)
  \end{eqnarray*}
  The isomorphism in the first step above is by the refinement process
  from \autoref{sec:refining-inductive}, and the equality in the
  second is by definition of $\Psi$. Now, to apply
  \thmref{thm:alg-adjunctions} we must show that $F^{\gamma,\delta}
  \circ \Psi \cong \Psi \circ F^{\overline{\gamma,\delta}}$. So
  suppose $X$ is in $\Set^{D \times A}$. Then
  \begin{eqnarray*}
    & & F^{\gamma,\delta}(\Psi(X)) \\
    &=& (\lambda a.\, \{x : F\{(D \times A,\Phi(\Psi(X)))\} \sepbar
    \gamma(F\pi x) = a \},\, 
    \lambda (a,x).\, \delta(F\pi_1(F\pi x)) \}) \\
    &\cong& (\lambda a.\, \{x : F\{(D \times A, X)\} \sepbar
    \gamma(F\pi x) = a \},\, \lambda (a,x).\, \delta(F\pi_1(F\pi x)) \})
 \end{eqnarray*}
 Here, we have used the fact that the functors $\Phi$ and $\Psi$ form
 an adjoint equivalence by \thmref{thm:ir-equiv}. On the other hand,
  \begin{eqnarray*}
    & & \Psi (F^{\overline{\gamma,\delta}} X) \\
    &=& \Psi (\lambda (d,a).\, \{x : F\{(D \times A,X)\} \sepbar
    \overline{(\gamma,\delta)}(F \pi x) = (d,a) \}) \\
    &=& \Psi (\lambda (d,a).\, \{x : F\{(D \times A,X)\} \sepbar \gamma(F \pi
    x) = a,\, \delta(F\pi_1(F\pi x)) = d \}) \\
    &\cong& (\lambda a.\, \{(d : D, x : F\{(D \times A,X)\}) \sepbar \gamma(F
    \pi x) = a,\, \delta(F \pi_1 (F \pi x)) = d\},\, \lambda (a, (d,x)).\, d) \\
    &\cong& (\lambda a.\, \{x : F\{(D \times A,X)\} \sepbar \gamma(F \pi x) =
    a\},\, \lambda (a, x).\, \delta(F \pi_1 (F \pi x)))
  \end{eqnarray*}
  by the definition of $\overline{\gamma,\delta}$. So, by the comment
  after \thmref{thm:alg-adjunctions}, $\Psi(\mu
  F^{\overline{\gamma,\delta}}) \cong \mu F^{\gamma,\delta}$. But
  since $\Psi(\mu F^{\overline{\gamma,\delta}})$ is the same
  as \parenref{eq:zygo-refine-spec}, we have
  that \parenref{eq:zygo-refine-spec} can indeed be inductively
  characterised as $\mu F^{\gamma,\delta}$.
\end{proof}

It is also possible to state and prove a generalisation of
\thmref{thm:zygo-refine-correct} in the general setting of a full
cartesian Lawvere category with very strong coproducts, as defined in
\autoref{sec:families}. In this case, we make use of the category
$\mathcal{E}_A \times_{\mathcal{B}} \mathcal{B}/D$, which is defined
by a pullback construction similar to that in
\autoref{sec:initial-algebra-ir}. The use of very strong coproducts is
essential to proving the generalised analogue of the adjoint
equivalence in \thmref{thm:ir-equiv}. In the general fibrational
setting, we have the following definition of $F^{\gamma,\delta}$:
\begin{displaymath}
  F^{\gamma,\delta}(X,f) = (\Sigma_f(\hat{F}_{D \times A}(\Phi(X,f))), \delta \circ F\pi_1 \circ F\pi_{\Phi(X,f)})
\end{displaymath}

The formulation of zygomorphic refinement in the general setting of a
full cartesian Lawvere category with very strong coproducts means that
we can use the process described in \autoref{sec:indexed-refinement}
to derive a fibration in which to perform zygomorphic refinement on
indexed inductive types.

\begin{eorollary}
  \exref{ex:factorial-refine} illustrates refinement by a
  paramorphism, but does not use the full generality of refinement by
  a zygomorphism. We now demonstrate the power of refinement by a
  zygomorphism to mechanically derive an inductive characterisation of
  the data type of lists of rational numbers indexed by their average.

  We specialise the functor $F_{\tyname{List}_B}$ from
  \exref{ex:lists-n-vectors} to get the functor representing the type
  of lists of rational numbers: $F_{\tyname{List}_{\mathbb{Q}}}X = 1 +
  \mathbb{Q} \times X$. We reuse the $F_{\tyname{List}_B}$-algebra
  $\mathit{lengthalg} : F_{\tyname{List}_B} \mathbb{N} \to
  \mathbb{N}$, also from \exref{ex:lists-n-vectors}, whose fold
  computes the length of a list. We also consider the following
  $F_{\tyname{List}_{\mathbb{Q}}}$-algebra $\mathit{sumalg}$, which is
  used to compute the sum of the elements of a list:
  \[\begin{array}{lll}
    \mathit{sumalg} & : & F_{\tyname{List}_{\mathbb{Q}}}\mathbb{Q} \to
\mathbb{Q} \\
    \mathit{sumalg}\ \mathsf{Nil} & = & 0 \\
    \mathit{sumalg}\ (\mathsf{Cons}(q,s)) & = & q + s
  \end{array}\]
  By the standard construction of the product of two $F$-algebras, we
  combine $\mathit{lengthalg}$ and $\mathit{sumalg}$ to produce the
  following single $F_{\tyname{List}_{\mathbb{Q}}}$-algebra whose fold
  will simultaneously compute the sum and length of a list of rational
  numbers: 
  \begin{displaymath}
    \mathit{sumlengthalg} :
    F_{\tyname{List}_{\mathbb{Q}}}(\mathbb{Q} \times \mathbb{N}) \to
    \mathbb{Q} \times \mathbb{N}
  \end{displaymath}
  This algebra will form the $F$-algebra component of the zygomorphism
  by which we will refine $\mu F_{\tyname{List}_\mathbb{Q}}$.

  The morphism component of the zygomorphism by which we will refine
  $\mu F_{\tyname{List}_\mathbb{Q}}$ has carrier $1 +
  \mathbb{Q}$. Here, the non-$\mathbb{Q}$ case caters for empty lists,
  for which the average is not defined. We use $\mathsf{empty}$ and
  $\mathsf{avg}$ as mnemonics for the left and right injections into
  $1 + \mathbb{Q}$. The morphism $\mathit{avg}$ is defined by
 \[\begin{array}{lll}
    \mathit{avg} & : & F_{\tyname{List}_{\mathbb{Q}}}((\mathbb{Q} \times
    \mathbb{N}) \times (1 + \mathbb{Q})) \to 1 + \mathbb{Q} \\
    \mathit{avg}\ \mathsf{Nil} & = & \mathsf{empty} \\
    \mathit{avg}\ (\mathsf{Cons}(q,((s,l),\_))) & = & \mathsf{avg}(\frac{q +
      s}{l + 1})
  \end{array}\]

  Following a similar process to that in
  \exref{ex:factorial-refine}, we can now compute the refinement of
  $\mu F_{\tyname{List}_{\mathbb{Q}}}$ by $\mathit{sumlengthalg}$ and
  $\mathit{avg}$:
  \begin{displaymath}
    \begin{array}{l}
      F_{\tyname{List}_{\mathbb{Q}}}^{\mathit{avg},\mathit{sumlengthalg}}(X,f) = \\
      \quad
      \begin{array}{l}
       ( \lambda a.\ \{ * \sepbar a = \mathsf{empty} \} + \{ (q,a',x : Xa')
\sepbar a = \mathsf{avg}(\frac{q + \pi_1(f(a',x))}{\pi_2(f(a',x)) + 1}) \}, \\
        \lambda (a,x).\ \mathsf{case}\ x\ \mathsf{of}
        \left\{
          \begin{array}{lcl}
            \mathsf{Nil} & \Rightarrow & (0,0) \\
            \mathsf{Cons}(q,a',x) & \Rightarrow & (q + \pi_1(f(a',x)),
\pi_2(f(a',x)) + 1))
          \end{array}
        \right.
      \end{array}
    \end{array}
  \end{displaymath}
  In this definition, we have used $\pi_1(f(a',x))$ to obtain the sum
  of the list underlying $x$, and have likewise used $\pi_2(f(a',x))$
  to obtain its length. Expressing this refinement in Agda 2 gives the
  following definition:
\begin{verbatim}
mutual
  data AvgList : 1 + Rational -> Set where
    nil  : AvgList empty
    cons : (q : Rational) ->
           {a : 1 + Rational} ->
           (x : AvgList a) ->
           AvgList (avg ((q + sum x) / (length x + 1)))

  sum : {a : 1 + Rational} -> AvgList a -> Rational
  sum nil        = 0
  sum (cons q x) = q + sum x

  length : {a : 1 + Rational} -> AvgList a -> Nat
  length nil        = 0
  length (cons q x) = length x + 1
\end{verbatim}

\end{eorollary}

The fact we have generated small indexed inductive-recursive types by
a process of refinement by a zygomorphism leads to the interesting
question of whether it is possible to further refine small indexed
inductive-recursive types by any sort of refinement process. A
thorough investigation of such processes should also involve large
induction-recursion (recall that a large inductive-recursive type
entails the definition of a $\Set$-valued recursive function
simultaneously with the inductive type). The setting of large
inductive-recursive types is much more complicated than small
(indexed) inductive-recursive types, and so we leave investigation of
the refinement of general inductive-recursive types to future
work. Recent work by Malatesta, Altenkirch, Ghani, Hancock and McBride
\cite{malatesta12small} has shown that a large universe of small
inductive-recursive types described by codes is equivalent to the
universe of indexed containers \cite{alten09indexed}. This work may
point to a way to formulate the development of this section in terms
of codes for functors describing types rather than directly in terms
of the functors themselves.

Another interesting avenue for future work is to determine whether the
partial refinement process of \autoref{sec:partial-refinement} can be
combined with the zygomorphic refinement process presented in this
section.

\section{Conclusions, Applications, Related and Future Work}
\label{sec:discussion}

We have given a clean semantic framework for deriving refinements of
inductive types that store computationally relevant information within
the indices of the resulting refined types. We have also shown how
already indexed types can be refined further, how refined types can be
derived even when some elements of the original type do not have
indices, and how refinement by zygomorphisms entails the use of small
indexed induction-recursion for information hiding. In addition to its
theoretical clarity, the theory of refinement we have developed has
potential applications in the following areas:

\medskip

\noindent {\em Dependently Typed Programming:} Often a user is faced
with a choice between building properties of elements of data types
into more sophisticated data types, or stating these properties
externally as, say, pre- and post-conditions. While the former is
clearly preferable because properties can then be statically
type-checked, it also incurs an overhead which can deter its
adoption. Supplying the programmer with infrastructure to produce
refined types as needed can reduce this overhead.

\vspace*{0.03in}

\noindent {\em Libraries:} With the implementation of refinement,
library implementers will no longer need to provide comprehensive
collections of data types, but instead only methods for defining new
data types. Our results also ensure that library implementers will not
need to guess which refinement types will prove useful to programmers,
and can instead focus on providing useful abstractions for creating
more sophisticated data types from simpler ones.

\vspace*{0.03in}

\noindent {\em Implementation:} Current implementations of types such
as {\tt Vector} types store all index information. For example, a
vector of length 3 will store the lengths 3, 2, and 1 of its
subvectors. Since this can be very space-consuming, Brady \emph{et
  al.}~\cite{brady03inductive} have sought to determine when this
information need not be stored in memory. Our work suggests that a
refinement $\mu F^{\alpha}$ can be implemented by simply implementing
the underlying type $\mu F$, since programs requiring indices can
reconstruct these as needed. It could therefore provide a
user-controllable tradeoff between space and time efficiency.

\subsection{Related Work}

The work closest to that reported here is McBride's work on ornaments
\cite{mcbride10ornaments}. McBride defines a type of descriptions of
inductive data types, along with a notion of one description
``ornamenting'' another. Despite the differences between our
fibrational approach and his type-theoretic approach, the notion of
refinement presented in \hyperref[sec:refining-inductive]{Sections
  \ref*{sec:refining-inductive}} and
\hyperref[sec:indexed-refinement]{ \ref*{sec:indexed-refinement}} is
very similar to McBride's notion of an algebraic
ornament. Ornamentation further allows for additional arbitrary data
to be attached to constructors, something that is not possible with
any of the refinement processes that we have discussed in this
paper. On the other hand, ornamentation is restricted to inductive
types and so does not allow for the generation of indexed
inductive-recursive types that we presented in
\autoref{sec:zygo-refine}. The theory of ornamentation has been
developed by Ko and Gibbons \cite{KoGibbons2011OAOAOO}, who examine
the relationship between the ornamental versions of the ``local'' and
``global'' refinement that we discussed in
\autoref{sec:better-solution}. More recently, Dagand and McBride
\cite{DagandMcBride2012funOrn} have described an extension of
McBride's original definition of ornamentation which allows for the
removal of constructors. In our setting, the removal of constructors
is possible with the use of partial refinement
(\autoref{sec:partial-refinement}).

An interesting question for future work is to determine the
relationship between functions defined on data types and functions
defined on refined versions of data types. This question has been
addressed in the setting of McBride's work on ornaments by Ko and
Gibbons \cite{KoGibbons2011OAOAOO} and also by Dagand and McBride
\cite{DagandMcBride2012funOrn}. We have not considered the question of
refinement of functions in this paper, and we leave it as future work
to determine whether or not the fibrational approach taken here can
provide any insight.

Chuang and Lin \cite{chuang06algebra} present a way to derive new
indexed inductive types from existing inductive types and algebras
that is very similar to our basic refinement process in
\autoref{sec:refining-inductive}. Chuang and Lin work in the setting
of the codomain fibration, which makes some calculations easier, but
extensions to partial and zygomorphic refinement more difficult.

A line of research allowing the programmer to give refined types to
constructors of inductive data types was initiated by Freeman and
Pfenning \cite{freeman91refinement}. Freeman and Pfenning defined a
variant of ML that allowed programmers to define refinements of
inductive types by altering the types of constructors, or by
disallowing the use of certain constructors. Refinement of this sort
did not require dependent types. This work was later developed by Xi
\cite{xi00dependently}, Davies \cite{davies05practical} and Dunfield
\cite{dunfield07unified} for extensions of ML-like languages with
dependent types, and by Pfenning \cite{pfenning93refinement} and Lovas
and Pfenning \cite{lovas10refinement} for LF. The work of Kawaguchi
\emph{et al.}  \cite{kawaguchi09type} is also similar. This research
begins with an existing type system and provides a mechanism for
expressing richer properties of values that are well-typeable in that
type system.  It is thus similar to the work reported here, although a
major focus of the work of Freeman and Pfenning and its descendants is
on the decidability of type checking and inference of refined types,
which we have not considered in this paper. On the other hand, we
formally prove that each refinement is isomorphic to the richer,
property-expressing data type it is intended to capture, rather than
leaving this to the programmer to justify on a
refinement-by-refinement basis.

Refinement types have been used in other settings to give more precise
types to programs in existing programming languages (but not
specifically to inductive types). For example, Denney
\cite{denney98refinement} and Gordon and Fournet
\cite{gordon09principles} use subset types to refine the type systems
of ML-like languages. Subset types are also used heavily in the PVS
theorem prover \cite{rushby98subtypes}.

Our results extend the systematic code reuse delivered by generic
programming \cite{amm07,bghj07,bdj03}: in addition to generating new
programs we can also generate new types from existing types. This area
is being explored in Epigram \cite{chapman10gentle}, with codes for
data types being represented within a predicative intensional
system. This enables programs to generate new data types. It should be
possible to implement our refinement process using similar techniques.

In addition to the specific differences between our work and that
discussed above, a distinguishing feature of ours is the semantic
methodology we use to develop refinement. We believe that this
methodology is new. We also believe that a semantic approach is
important: it can serve as a principled foundation for refinement, as
well as provide a framework in which to compare different
implementations. Moreover, it may lead to new algebraic insights into
refinement that complement the logical perspective of previous work.

\bigskip

{\em Acknowledgements:} We thank Conor McBride, Frank Pfenning, and
Pierre-Evariste Dagand for helpful comments on this work. The
anonymous FoSSaCS and LMCS reviewers also provided useful
feedback. This work was funded by EPSRC grant EP/G068917/1.

\bibliographystyle{plain}

\bibliography{refinement-journal}

\end{document}